\newif\ifdraft \draftfalse
\newif\iffull \fulltrue

\makeatletter \@input{tex.flags} \makeatother

\documentclass[11pt]{article}

\usepackage{algorithm}
\usepackage[noend]{algorithmic}

\usepackage{fullpage}

\usepackage{amsmath, amssymb, amsthm}
\usepackage{mathtools}
\usepackage{verbatim}
\usepackage{color}
\usepackage{times}
\usepackage{xcolor}

\usepackage{multicol}

\definecolor{DarkGreen}{rgb}{0.1,0.5,0.1}
\definecolor{DarkRed}{rgb}{0.5,0.1,0.1}
\definecolor{DarkBlue}{rgb}{0.1,0.1,0.5}
\usepackage[]{hyperref}
\hypersetup{
    unicode=false,          
    pdftoolbar=true,        
    pdfmenubar=true,        
    pdffitwindow=false,      
    pdftitle={},    
    pdfauthor={}
    pdfsubject={},   
    pdfnewwindow=true,      
    pdfkeywords={keywords}, 
    colorlinks=true,       
    linkcolor=DarkRed,          
    citecolor=DarkGreen,        
    filecolor=DarkRed,      
    urlcolor=DarkBlue,          
}
\usepackage{xspace}
\usepackage{cleveref}
\usepackage{thmtools, thm-restate}

\iffull
\usepackage[numbers]{natbib}
\else
\usepackage[numbers,sort&compress]{natbib}
\fi

\newcommand{\sw}[1]{\ifdraft \textcolor{blue}{[Steven: #1]}\fi}
\newcommand{\ar}[1]{\ifdraft \textcolor{red}{[Aaron: #1]}\fi}
\newcommand{\jr}[1]{\ifdraft \textcolor{brown}{[Jaikumar: #1]}\fi}

\newcommand\NN{\mathbb{N}}
\newcommand\RR{\mathbb{R}}
\newcommand\cA{\mathcal{A}}

\newcommand\cF{\mathcal{F}}

\newcommand\cM{\mathcal{M}}

\newcommand\cR{\mathcal{R}}
\newcommand\cL{\mathcal{L}}
\newcommand\cX{\mathcal{X}}

\newcommand\cS{\mathcal{S}}

\newcommand{\KL}{D_{\mathrm{KL}}}
\newcommand{\ui}{^{(i)}}
\newcommand{\ranG}{{\sf RanG}}
\newcommand{\score}{{\sf score}}
\newcommand{\stab}{{\sf Stab}}
\newcommand{\adm}{{\sf admit}}
\newcommand{\temp}{{\sf temp}}
\newcommand{\br}{{\sf BR\text{-}Sim}}
\newcommand{\counter}{{\sf Counter}}

\newcommand{\exflow}{{\sf ExtractPath}}
\newcommand{\reg}{{\sf ReC}}
\newcommand{\approxC}{{\sf ApproxCount}}
\newcommand{\exstream}{{\sf ExtractCount}}
\newcommand{\pritocol}{{\sf PriCoor}}

\renewcommand{\tilde}{\widetilde}

\newcommand{\ub}{^\bullet}

\DeclareMathOperator*{\Expectation}{\mathbb{E}}
\newcommand{\Ex}[2]{\Expectation_{#1}\left[#2\right]}

\newcommand{\eps}{\varepsilon}
\def\epsilon{\varepsilon}

\DeclareMathOperator{\OPT}{OPT}
\renewcommand{\hat}{\widehat}

\DeclareMathOperator*{\argmax}{\mathrm{argmax}}

\newcommand{\INDSTATE}[1][1]{\STATE\hspace{#1\algorithmicindent}}
\newtheorem*{theorem*}{Theorem}

\declaretheorem[
  name=Theorem,
  refname={theorem, theorems},
  Refname={Theorem, Theorems}]{theorem}
\declaretheorem[
  name=Lemma,
  refname={lemma, lemmas},
  Refname={Lemma, Lemmas}]{lemma}
\declaretheorem[
  name=Fact,
  refname={fact, facts},
  Refname={Fact, Facts}]{fact}
\declaretheorem[
  name=Claim,
  refname={claim, claims},
  Refname={Claim, Claims}]{claim}
\declaretheorem[
  name=Remark,
  refname={remark, remarks},
  Refname={Remark, Remarks}]{remark}
\declaretheorem[
  name=Corollary,
  refname={corollary, corollaries},
  Refname={Corollary, Corollaries}]{corollary}
\declaretheorem[
  name=Definition,
  refname={definition, definitions},
  Refname={Definition, Definitions}]{definition}

\title{Coordination Complexity:\\
Small Information Coordinating Large Populations}

\author{
Rachel Cummings\footnotemark[1]
\and
Katrina Ligett\footnotemark[1]
\and
Jaikumar Radhakrishnan\footnotemark[2]
\and
Aaron Roth\footnotemark[3]
\and
Zhiwei Steven Wu\footnotemark[3]
}

\begin{document}

\maketitle
\renewcommand{\thefootnote}{\fnsymbol{footnote}}
\footnotetext[1]{Computing and Mathematical Sciences, California Institute of Technology; Email: \texttt{\{rachelc, katrina\}@caltech.edu}; Supported in part by
NSF grant CNS-1254169, US-Israel Binational Science Foundation grant
2012348, the Charles Lee Powell Foundation, a Google Faculty Research
Award, an Okawa Foundation Research Grant, a Microsoft Faculty
Fellowship, and a Simons Award for Graduate Students in Theoretical
Computer Science.}
\footnotetext[2]{School of Technology and Computer Science, Tata Institute of Fundamental Research; A portion of this work was done while the author was visiting the Simons Institute for Theory of Computing in Berkeley, CA. Email: \texttt{jaikumar@tifr.res.in}}
\footnotetext[3]{Computer and Information Science, University of Pennsylvania;
Email:\texttt{\{aaroth,wuzhiwei\}@cis.upenn.edu}; Supported in part by NSF Grant CCF-1101389, an NSF CAREER award, and an Alfred P. Sloan Foundation Fellowship.
}
\renewcommand{\thefootnote}{\arabic{footnote}}

\thispagestyle{empty}
\begin{abstract}
We study a quantity that we call \emph{coordination
complexity}. In a distributed optimization problem, the information
defining a problem instance is distributed among $n$ parties, who need
to each choose an action, which jointly will form a solution to the
optimization problem. The coordination complexity represents the
minimal amount of information that a centralized coordinator, who has
full knowledge of the problem instance, needs to broadcast in order to
coordinate the $n$ parties to play a nearly optimal solution.

We show that upper bounds on the coordination complexity of a problem
imply the existence of good jointly differentially private algorithms
for solving that problem, which in turn are known to upper bound the
price of anarchy in certain games with dynamically changing
populations.

We show several results. We fully characterize the coordination
complexity for the problem of computing a many-to-one matching in a
bipartite graph by giving almost matching lower and upper
bounds.\sw{changed} Our upper bound in fact extends much more
generally, to the problem of solving a linearly separable convex
program. We also give a different upper bound technique, which we use
to bound the coordination complexity of coordinating a Nash
equilibrium in a routing game, and of computing a stable matching.
\end{abstract}

\newpage
\clearpage
\setcounter{page}{1}

\section{Introduction}
In this paper, we study a quantity which we call
\emph{coordination complexity}. This quantity measures the amount of
information that a centralized coordinator needs to broadcast in order
to coordinate $n$ parties, each with only local information about a
problem instance, to jointly implement a globally optimal
solution. Unlike in \emph{communication complexity}, there is no need
for the communication protocols in our setting to derive the optimal
solution starting with nothing but local information, nor even
\emph{verify} that a proposed solution is optimal (as is the goal in
non-deterministic communication complexity). Instead, in our setting,
there is a central coordinator who already has complete knowledge of
the problem instance --- and hence also of the optimal solution. His
goal is simply to publish a concise message to guide the $n$ parties
making up the problem instance to coordinate on the desired solution
--- ideally using fewer bits than would be (trivially) needed to simply
publish the optimal solution itself.\footnote{Within our framework of
  coordination complexity, we assume that the players are not
  strategic --- they will faithfully follow the
  coordination protocol upon observing the message broadcast by the
  coordinator. We do study the interface between the coordination
  complexity and the strategic variants of some problems
  in~\Cref{s.privacy}.\sw{added} }

Aside from its intrinsic interest, our motivation for studying this quantity is two-fold. First, as we show, problems with low coordination complexity also have good protocols for implementing nearly optimal solutions under the constraint of \emph{joint differential privacy} \cite{DMNS06,KPRU14} --- i.e. protocols that allow the joint implementation of a nearly optimal solution in a manner such that no coalition of parties can learn much about the portion of the instance known by any party not in the coalition. The existence of jointly differentially private protocols in turn have recently been shown to imply a low ``price of anarchy'' for no-regret players in the strategic variant of the optimization problem when the game in question is smooth --- even when the population is dynamically changing \cite{dynamicPOA}. Hence, as a result of the connection we develop in this paper, in order to show dynamic price of anarchy bounds of the sort given in \cite{dynamicPOA}, it is sufficient to show that the game in question has low coordination complexity, without needing to directly develop and analyze differentially private algorithms. Using this connection we also derive new results for what can be implemented under the constraint of pure joint differential privacy --- results that were previously only known subject to approximate joint differential privacy.

Second, coordination complexity is a stylized measure of the power of
concise broadcasts (e.g. prices in the setting of allocation problems,
or congestion information in the setting of routing problems) to
coordinate populations in the absence of any interaction.\footnote{Of
  course, the connection here is in a stylized model --- in a market,
  there is not in fact any party with complete information of the
  problem instance --- but the market is nevertheless encoding good
  ``distributional information'' about the population of buyers likely
  to arrive.} Here we note that prices seem to coordinate markets,
despite the fact that individuals do not actually participate in any
kind of interactive ``Walrasian mechanism'' of the sort that would be
needed to compute the allocation itself, in addition to the prices
(see e.g. \cite{KC82,DNO14}). Indeed, prices alone are generally not
sufficient to coordinate high welfare allocations because prices on
their own can induce a large number of indifferences that might need
to be resolved in a coordinated way --- and hence Walrasian equilibria
are defined not just as vectors of equilibrium prices, but as vectors
of prices paired with optimal allocations. Publishing a Walrasian
equilibrium would be a trivial solution in our setting, because it
involves communicating the entire solution that we wish to coordinate
--- the optimal allocation. Nevertheless, we show that the
coordination complexity of the allocation problem is --- up to log
factors --- equal to the number of types of goods in a commodity
market. This is the same as what would be needed to communicate prices
(indeed, our solution can be viewed as communicating prices in a
slightly different, ``regularized'' market), and can be substantially
smaller than what would be needed to communicate the optimal
allocation itself.

\subsection{Our Results and Techniques}
In our model (which we formally define in Section \ref{sec:prelims}), a \emph{problem instance} $D$ is defined by an $n$-tuple from some abstract domain $\cX$: $D \in \cX^n$. We write $D = (D^{(1)},\ldots,D^{(n)})$ to denote the fact that the information defining the problem instance is partitioned among $n$ agents, and each agent $i$ knows only his own part $D^{(i)}$. The solution space is also a product space: $\cA^n$, and each agent $i$ can choose a single \emph{action} $a_i \in \cA$ -- the choices of all of the agents jointly form a solution $a = (a_1,\ldots,a_n)$. The \emph{coordinator} knows the entire problem instance $D$, and publishes a signal $\sigma(D) \in \{0,1\}^\ell$. Each agent then chooses an action $a_i := \pi(D^{(i)}, \sigma(D))$ based only on the coordinator's signal and her own part of the problem instance. The jointly induced solution $a = (a_1,\ldots,a_n)$ is the output of the interaction. The pair of functions $\sigma, \pi$ jointly form a protocol, and $\ell$, the length of the coordinator's signal is the coordination complexity of the protocol. The coordination complexity of a problem is the minimal coordination complexity of any protocol solving the problem.

A canonical example to keep in mind is many-to-one matchings: Here, a problem instance is defined by a bipartite graph between $n$ agents and $k$ types of goods. Each good $j$ has a supply $s_j$, and the goal is to find a maximum cardinality matching such that no agent is matched to more than one good, and no good $j$ is matched to more than $s_j$ agents. Here, the portion of the instance known to agent $i$ is the set of goods adjacent to agent $i$-- but nothing about the goods adjacent to other agents. Note that describing a matching requires $\Omega(n \log k)$ bits, which is the trivial upper bound on the coordination complexity for this problem. For this problem, we show nearly matching upper and lower bounds: no protocol with coordination complexity $o(k)$ can guarantee a constant approximation to the optimal solution, whereas there is a protocol with coordination complexity $O(k\log n)$ that can obtain a $(1 + o(1))$-approximation to the optimal solution. Our upper bound in fact extends much more generally, to any problem that can be written down as a convex program whose objective and constraints are linearly separable between agents' data.

 The idea of the upper bound is to broadcast a portion of the optimal dual solution to the convex program -- one dual variable for every constraint that is defined by the data of multiple agents (there is no need to publish the dual variables corresponding to constraints that depend only on the data of a single agent). For the many-to-one matching problem, these dual variables correspond to ``prices'' -- one for each of the $k$ types of goods. This idea on its own does not work, however, because a dual optimal solution to a convex program is not generally sufficient to specify the primal optimal solution. When specialized to the case of matchings, this is because optimal ``market clearing prices'' can induce a large number of indifferences among goods for each of the $n$ agents, and these indifferences might need to be broken in a coordinated way to induce an optimal matching. To solve this problem, we instead release the dual variables corresponding to a slightly different convex program, in which a \emph{strongly convex regularizer} has been added to the objective. The effect of the strongly convex regularizer is that the optimal dual solution now uniquely specifies the optimal primal solution -- although now the optimal primal solution to a modified problem. The rest of our approach deals with trading off the weight of the regularizer with the number of bits needed to approximately specify each of the dual variables, and the error of the regularized optimal solution relative to the optimal solution to the original problem.

We also give several other positive results, based on a different
technique: broadcasting the truncated transcript of a process known to
converge to a solution of interest. Using this technique, we give low
coordination complexity protocols for the problem of coordinating on
an equilibrium in a routing game, and for the problem of
coordinating on a \emph{stable} many-to-one matching.

Finally, we show that problems that have both low sensitivity objectives (as all of the problems we study do) and low coordination complexity also have good \emph{jointly differentially private} protocols. Using the results of \cite{dynamicPOA}, this also shows a bound on the price of anarchy of the strategic variant of these problems, whenever they are \emph{smooth games}, which holds even under a dynamically changing population.

\subsection{Related Work}

Our model of coordination complexity is related to, but distinct from,
the well-studied notion of communication complexity --- see
\cite{KN97} for a textbook introduction.  While both complexity
notions measure the number of bits that must be transmitted among
decentralized parties to reach a particular outcome, they differ in
the initial endowment of information, as well as in the requirements
of each player to know the final outcome.  In communication
complexity, the information describing the problem instance is fully
distributed, and communication is necessary for all parties to know
the outcome. Coordination complexity in contrast assumes the existence
of a coordinator who knows the entire problem instance, and must
broadcast information to the players which will allow them to each
compute their part of the output -- there is no need for any of the
parties to know the entire output. More similar to our setting is
\emph{non-deterministic communication complexity}, in which we may
imagine that there is an oracle who knows the inputs of all players
and broadcasts a message (perhaps partially) describing a solution
together with a certificate that allows the parties to verify the
optimality of the solution. In contrast, in our model of coordination
complexity, the coordinator does not need to provide any certificate
allowing parties to verify that the coordinated solution is optimal
(indeed, each party need not have any information about the portion of
the solution proposed to other parties).



The informational requirements of coordinating matchings has a long
history of study in economics, and has recently gained attention in
theoretical computer science.  Hayek's classic paper \cite{Hay45}
conjectured that Walrasian price mechanisms, which coordinate
matchings via a t\^atonnement process that updates market-clearing
prices based on demand, are ``informationally efficient,'' in that
they verify optimal allocations with the least amount of information.
This was later formalized by \cite{Hur60} and \cite{MR74} in specific
settings of interest, using an informational metric that measured
smooth real-valued communication. Nisan and Segal study the communication complexity of matchings using the tools of communication complexity as developed in computer science, and show that any communication protocol that determines an optimal allocation must also determine supporting prices \cite{NS06}.  Recently, \cite{DNO14} and
\cite{ANRW15} studied the problem of computing an optimal matching through the lens of interactive
communication complexity, showing that interactive protocols can have significantly
lower communication complexity than non-interactive ones. Note that the communication complexity bounds given in these papers are always larger than the description length of the matching itself -- in contrast, here when we study coordination complexity, nontrivial bounds must not just be smaller than the input, but must also be smaller than the size of the optimal matching.

Finally, there are two papers that study a very similar setting to ours, although they obtain rather different results.\footnote{We thank Ilya Segal for pointing out \cite{Cal84} to us, and thank Sepehr Assadi for pointing out Section 4 of \cite{DPS02}.} Calsamiglia \cite{Cal84} studies a real-valued communication model in which a central coordinator with full knowledge of the instance needs to broadcast a concise signal to coordinate an allocation in an exchange market---see  \cite{SegSurvey} for context on how this result fits into the economic literature on communication complexity. Deng, Papadimitriou, and Safra also study a similar model in Section 4 of \cite{DPS02}, which they call ``Market Communication''. Despite the similarity in models, the results of both \cite{Cal84} and \cite{DPS02} stand in sharp contrast to ours---they both give \emph{lower bounds}, showing that the amount of communication necessary needs to grow linearly with the number of buyers $n$, while we give upper bounds showing that it is necessary to grow only with the number of different types of goods $k$. Calsamiglia does not allow approximation in his model, which is necessary for our results. Deng, Papadimitriou, and Safra allow for approximation, but study an instance of a problem that cannot be expressed as a linearly separable convex program, which shows that structure of the sort that we use is necessary.

A line of work \cite{KPRU14, RR14, CKRW14, HHRRW14, HHRW14, RRUW15}
has studied protocols for implementing outcomes in various settings
under the constraint of joint differential privacy \cite{DMNS06,
  KPRU14}, which allows $n$ parties to jointly implement some solution
while ensuring that no coalition of parties can learn much about the
input of any party outside the coalition.  Most (but not all) of these
algorithms are actually private coordination protocols of the sort we
study here, in which the algorithm can be viewed as a coordinator who
is constrained to broadcast a private signal. These jointly private
algorithms are not constrained to transmit a short signal -- and
indeed, the private signals can sometimes be verbose. But as we show,
problems with low coordination complexity also have good jointly
differentially private algorithms, which was one of our original
motivations for studying this quantity.

Lykouris, Syrgkanis, and Tardos \cite{dynamicPOA} show that the existence of a jointly
differentially private algorithm for solving an optimization problem
implies that the strategic variant of the problem has a low ``price of
anarchy'' for learning agents, even in dynamic settings, in which player types change over time, as long as the game is smooth.  Because we show
in Section \ref{s.privacy} that any problem with a low sensitivity
objective and low coordination complexity has a good jointly
differentially private algorithm, using the results of \cite{dynamicPOA}, to prove a bound on the price of anarchy in a smooth dynamic game, we show it suffices to bound the coordination complexity of the game.

\section{Preliminaries}
\label{sec:prelims}
A \emph{coordination problem} is defined by a set of $n$ agents, a
data domain $\cX$, an action range $\cA$, and a social objective
function $S\colon \cX^n \times \cA^n \rightarrow \RR$. An
\emph{instance} of a coordination problem consists of a set of $n$
elements from the data domain: $D = (D^{(1)}, \ldots, D^{(n)}) \in
\cX^n$. Each agent $i$ has knowledge only of $D^{(i)}$, his own
portion of the problem instance, and the goal is for a centralized
\emph{coordinator} to broadcast a concise message to the agents to
allow them to arrive at a solution $a = (a_1,\ldots,a_n) \in \cA^n$
that approximately maximizes the objective function $S(D, a)$.

A \emph{coordination protocol} consists of two functions, an
\emph{encoding function} $\sigma \colon \cX^n \rightarrow \{0, 1\}^*$
and a \emph{decoding function} $\pi \colon \cX \times \{0,1\}^*
\rightarrow \cA$. A coordination protocol $(\sigma, \pi)$ proceeds in
two stages:

\begin{itemize}
\item First the coordinator broadcasts the message $\sigma(D)$ to all
  agents using the encoding function.
\item Then each agent selects an action $a_i$ on the basis of her own
  portion of the problem instance and the broadcast message, using the
  decoding function: $a_i := \pi(D^{(i)}, \sigma(D))$.
\end{itemize}

Both functions $\sigma$ and $\pi$ may be randomized. The approximation
ratio of a protocol is the ratio of the optimal objective value to the
expected objective value of the solution induced by the protocol, in
the worst case over problem instances.

\begin{definition}[Approximation Ratio]
A coordination protocol $(\sigma,\pi)$ obtains a $\rho$ approximation
to a problem if:
$$\max_{D \in \cX^n} \frac{\OPT(D)}{\Ex{a_1,\ldots,a_n}{S(D, a)}} \leq
\rho$$ where each $a_i = \pi(D^{(i)} , \sigma(D))$, and the expectation is
taken over the randomness of $\sigma$ and $\pi$.
\end{definition}

The coordination complexity of a protocol is the maximum number of
bits the encoding function broadcasts, in the worst case over problem
instances.
\begin{definition}[Coordination Complexity]
  A coordination protocol $(\sigma, \pi)$ has \emph{coordination
    complexity} $\ell$ if:
    $$\max_{D \in \cX^n} |\sigma(D)| = \ell.$$
\end{definition}

The coordination complexity of obtaining a $\rho$ approximation to a
problem is the minimum value of the coordination complexity of all
protocols $(\sigma, \pi)$ that obtain a $\rho$ approximation to the
problem.

We conclude by making several observations about coordination
protocols. First, as we have defined them, they are
\emph{non-interactive} -- the coordinator first broadcasts a signal,
and then the agents respond. This is without loss of generality, since
the coordinator has full knowledge of the problem instance. Any
interactive protocol could be reduced at no additional communication
cost to a non-interactive protocol, simply by having the coordinator
publish the transcript that would have arisen from the interactive
protocol. This is in contrast to the setting of communication
complexity, in which interactive protocols can be more powerful than
non-interactive protocols (and makes it easier to prove lower bounds
for coordination complexity).

Second, the coordination complexity of a problem is trivially upper
bounded both by the description length of the problem instance (as is
communication complexity), \emph{and} by the description length of the
problem's optimal solution (unlike in non-deterministic communication
complexity, there is no need to pair the optimal solution with a
certificate allowing individual agents to verify it). Hence,
non-trivial bounds will be asymptotically smaller than both of these
quantities.

\paragraph{Bipartite Matching}{
The primary coordination problem we study in this paper is the
\emph{bipartite matching problem}. In this problem, there is a
bipartite graph $G = (V, W, E)$, in which every node in $V$ is
associated with a player and every node in $W$ represents a good.
Each player $i$'s private data is the set of edges incident to her
node -- i.e. $D^{(i)} = \{j : (i,j) \in E\}$. We study two variants of
this problem. In the \emph{one-to-one} matching problem, $W$
represents a set of distinct goods, and the goal is to coordinate a
maximum cardinality matching $E' \subseteq E$ such that for every $i
\in V$, $|\{j \in W: (i,j) \in E'\}| \leq 1$ and for every $j \in W$,
$|\{i \in V : (i,j) \in E'\}| \leq 1$. In the \emph{many-to-one}
matching problem, $W$ represents a set of $k$ commodities $j$, each
with a supply $b_j$.  The goal is to coordinate a maximum cardinality
many-to-one matching $E' \subseteq E$ such that for every $i \in V$,
$|\{j \in W: (i,j) \in E'\}| \leq 1$ and for every $j \in W$, $|\{i
\in V : (i,j) \in E'\}| \leq b_j$. The social objective in this
setting is the welfare or the cardinality of the matching, and we will
use $\OPT(G)$ to denote the optimal welfare objective.

Note that the resulting solution might not be feasible since the
players' demands are not always satisfied. We need to make sure that
we are not over-counting when measuring the welfare. In one-to-one
matchings, if more than one players select a good, only the first
player is matched to it. In many-to-one matchings, if more than $b_j$
players select a good of type $j$, only the first $b_j$ players are
matched the good $j$.}\sw{added}



\paragraph{Notation}{We use $\|\cdot \|$ to denote the $\ell_2$ norm,
 and more generally use $\|\cdot \|_p$ to denote the $\ell_p$ norm. }

\section{Lower Bounds for \iffull Bipartite \fi Matchings}
\label{sec:lower}
In this section, we present lower bounds on the coordination
complexity of bipartite matching problems. As a building block, we
prove a lower bound for the one-to-one matching problem on a bipartite
graph with $n$ vertices on each side, showing an $\Omega(n)$ lower
bound -- i.e. that no substantial improvement on the trivial solution
is possible. We then extend this lower bound to the problem of
many-to-one matchings, in which there are $n$ agents who must be
matched to $k$ goods (each good can be matched to many agents, up to
its supply). Here, we show an $\Omega(k)$ lower bound. In the next
section, we give a nearly matching upper bound, which substantially
improves over the trivial solution.

\subsection{A Variant of the Index Function Problem}
Before we present our lower bound, we introduce a variant of
the~\emph{random index function} problem~\cite{KNR99}, which will be
useful for our proof.

\paragraph{MULTIPLE-INDEX}{
There are two players Alice and Bob. Alice receives as input a
sequence of $t$ pairs, $I=\langle (S_i, u_i): i = 1,2, \ldots, t
\rangle$, where the $S_i$ are disjoint sets each with $k$ elements,
and $u_i$ is uniformly distributed in $S_i$. Based on her input Alice
sends Bob a message $M(I)$.  Bob then receives $(S_j,j)$, where $j$ is
chosen from $[t]$. Bob must determine $u_j$; let his output be
$B(S_j,j, M(I)) \in S_j$. We say that the protocol succeeds if $B(S_j,j,
M(I))=u_j$.  Let $\ell(t,k,p)$ be the minimum number of bits (for the
worst input) that Alice must send in order for Bob to succeed with
probability at least $p$.}

Note that if Bob guesses randomly, then the protocol already succeeds
with probability $p = 1/k$. The following result shows that any
significant improvement over this trivial probability of success will
require Alice to send Bob a long message. See \iffull appendix \else
the full version \fi for a full proof.

\begin{restatable}{lemma}{ranindex}\label{index}
For $p \geq 1/k$, we have $\ell(t,k,p) \geq (8\log e) t(p - 1/k)^2$.
\end{restatable}

\subsection{Lower Bound for One-to-One Matchings}\label{sec:1to1}
We will first focus on the lower bound on one-to-one matching and show
the following.
\begin{theorem}\label{1to1}
Suppose the coordination protocol $\Pi$ for one-to-one matching
guarantees an approximation ratio of $\rho$ in expectation. Then, the
coordinator of $\Pi$ must broadcast $\Omega(n/\rho^4)$ bits on problem
instances of size $n$ (in the worst case).
\end{theorem}

Fix the protocol $\Pi$. We will extract a two-party communication
protocol for the {\bf MULTIPLE-INDEX} problem from $\Pi$, and use the above
lemma. As a first step for our lower bound proof, we will consider the
following random graph construction process $\ranG$.

\paragraph{Random Graph Construction $\ranG(\rho, n)$:}{
Let $\kappa = \frac{n}{8\rho}$ and $A = \frac{n}{16\rho^2}$. 
Consider the following random bipartite graph $G$ with vertex
set $(V, W)$ such that $|V| = |W| = n$.
\begin{itemize}
\item Randomly generate an ordering $w_1,w_2, \ldots, w_n$ of $W$ (all
  $n!$ orderings being equally likely), and partition $W$ as
  $W_1 \cup W_2$ such that $W_1=\{w_1,w_2,\ldots,
  w_{\kappa}\}$, and $W_2=\{w_{\kappa+1}, w_{\kappa+2}, \ldots, w_n\}$.

\item Similarly, randomly generate an ordering $v_1, v_2, \ldots, v_n$
  of $V$, and parition $V$ into $n/A$ bocks, $B_1, B_2, \ldots, B_{n/A}$
  (each with $A$ vertices), where
\[ B_j := \{v_i: (j-1)A + 1  \leq i \leq j \, A\}.\]

\item Connect $B_j$ and $W$ as follows.  First, we describe the
  connections between $V$ and $W_1$.  The neighbourhoods of the
  vertices in each $B_j$ will be disjoint: we partition
  $W_1$ into equal-sized disjoint sets $(T_v: v \in B_j)$, and let the
  neighbours of $v \in B_j$ be exactly the $2\rho$ vertices in $T_v$.

\item In addition, assign each vertex in $v$ one neighbor in $W_2$, by
  connecting $V$ with $W_2$ in round-robin fashion --- connect vertex
  $v_i$ to vertex $w_j$, where $j= \left(\kappa+ i \mod (n-\kappa)\right)$.
  \jr{I have changed the expression; I hope I have not messed it up.}\sw{looks fine}
\end{itemize}
}

Before we prove~\Cref{1to1}, let us first observe that a graph
generated by $\ranG$ always has a matching with high welfare.

\begin{lemma}\label{goodgraph}
 Each graph $G$ generated by the above process $\ranG(\cdot, n)$ has
 optimal welfare $\OPT(G) \geq \frac{7n}{8}$.
 \ar{Have we defined ``The welfare objective''?}\sw{yes}
\end{lemma}
\iffull
\begin{proof}
Given the fixed ordering over the vertices in $V$, match each of its
first $(1 - \frac{1}{8\rho})n$ vertices to its unique neighbor in
$W_2$. Since $\rho\geq 1$, this gives a matching with welfare at least
$\left(1 - \frac{1}{8\rho}\right) n \geq \frac{7n}{8}$.
\end{proof}
\fi

\begin{proof}[Proof of~\Cref{1to1}]
Let $\Pi$ be a coordination protocol with coordination complexity
$\ell$ and approximation ratio $\rho$. This means that on a graph
instance generated by $\ranG$, the parties can coordinate on \jr{I
  don't understand $\ranG(\cdot,n)$.}\sw{fixed} a matching with
expected weight at least $\frac{7n}{8\rho}$.  Since $|W_1| =
\frac{n}{8\rho}$, we know in expectation at least $\frac{7n}{8\rho} -
\frac{n}{8\rho} = \frac{3n}{4\rho}$ of the vertices in $W_2$ are
matched. Let $\alpha_v$ be the probability that in $\Pi$, agent $v$
picks her neighbor in $W_2$. Then, by linearity of expectation,
$\sum_{v \in V} \alpha_v \geq \frac{3n}{4}$, that is, $\Expectation_{v
  \in V}[\alpha_v] \geq \frac{3}{4}$.  Hence, there must be some block
$B_j$ such that $\Expectation_{v \in B_j}[\alpha_v] \geq
\frac{3}{4}$. We will now restrict attention to the block $B_j$ and
consider the following instance of {\bf MULTIPLE-INDEX}: for each
$v\in B_j$, set $S_v = N(v)$---the neighborhood of vertex $v$, and let
$u_v$ be the vertex unique vertex in $S_v\cap W_2$. Since the message
broadcast by the coordination protocol allows the players to identify
the special element with average success probability of
$\frac{3}{4\rho}$, by~\Cref{index} the length of message
\begin{align*}
\ell  &\geq (8\log e) |B_j| \left(\frac{3}{4\rho} - \frac{1}{2\rho + 1}\right)^2\\
&\geq (8 \log e) \left(\frac{n}{16\rho^2}\right) \left( \frac{3}{4\rho} -
\frac{1}{2\rho + 1} \right)^2 
 \geq \Omega\left( \frac{n}{\rho^4}\right),
\end{align*}
which completes the proof.
\end{proof}

\subsection{Lower Bound for Many-to-One Matchings}
Finally, we give the following lower bound on coordination complexity
for many-to-one matchings. The lower bound relies on the result
from~\Cref{sec:1to1}---we show that any coordination protocol for
many-to-one matchings can also be reduced to a protocol for one-to-one
matchings, and so the lower bound in~\Cref{sec:1to1} can be extended
to give a lower bound for the many-to-one setting.

For our lower bound instance, we consider bipartite graphs $G = (V, E)$
such that the vertices in $V$ represent $n$ different players and $W$
represent a set of $k$ goods $j$, each with a supply $b$.

\jr{Should we say that in our lower bound instance, all $s_j=b$?}

\begin{theorem}\label{thm:many-to-one}
Suppose that there exists a coordination protocol for many-to-one
matchings that guarantees an approximation ratio of $\rho$ in
expectation. Then such a protocol has coordination complexity of
$\Omega(k / \rho^{4})$.
\end{theorem}

We will start by considering a one-to-one matching instance generated
by $\ranG$ with $k$ vertices on each side of the graph $G' = (V', W',
E')$.  By~\Cref{goodgraph}, the optimal matching of $G'$ has size
$\OPT' \geq \frac{7k}{8}$. \ar{Shouldn't we just assume that the
  $k\times k$ one-to-one matching instance is generated by the process
  that generated the lower bound instance in the section
  above?}\sw{fixed}

\jr{I agree. Let us just say that we start with the $k \times k$ instance
from the section above.}

Now we will turn this into an instance of a many-to-one matching
problem: make $b$ copies of each vertex in $V'$ to obtain vertex set
$V$, and set $W:= W'$ such that the supply of each good $j$ is $b$;
then for an edge $(v', w')$ in the original graph, connect all copies
of $v'$ to $w'$ in the new graph. This gives a bipartite graph $G =
(V, W, E)$. The following claim is straightforward.
\begin{claim}
The new graph $G$ has a matching of size at least $(b\OPT')$.
\end{claim}

Now suppose that we could coordinate the players in $V$ to obtain a
matching $M^*$ of size $\frac{b\OPT'}{\rho}$ in $G$. Then with a simple
sampling procedure, we can extract a high cardinality matching for
the original graph: for each vertex in $v' \in V'$, sample one of the
$b$ copies of $v'$ in $G$ uniformly at random along with its incident
matched edge. If two vertices in $V'$ are connected to the same type
of good in $W'$, break ties arbitrarily and keep only one of the
edges.

\begin{restatable}{lemma}{sampledmatching}\label{goodwelfare}
  The sampled matching in $G'$ has expected size at least
  $\frac{\OPT'}{3\rho}$.
\end{restatable}

We will defer the proof to the \iffull appendix\else full version\fi. We now have all the
pieces to prove~\Cref{thm:many-to-one}.
\begin{proof}[Proof of~\Cref{thm:many-to-one}]
Suppose that there exists a coordination protocol $(\sigma, \pi)$ for
many-to-one matchings with a guaranteed approximation ratio of
$\rho$. By the result of~\Cref{goodwelfare}, we know that this
coordination protocol for one-to-one matchings has an approximation
ratio most $O(\rho)$. By the lower bound in~\Cref{1to1}, we know that
the length of $\sigma(G)$ is at least $\Omega(k/\rho^{4})$.
\end{proof}

\section{Coordination Protocol for Linearly Separable Convex Programs}\label{sec:convex}
In this section, we give a coordination protocol for problems which
can be expressed as linearly separable convex programs, with
coordination complexity scaling only with the number of constraints
that bind between agents (so called \emph{coupling constraints},
defined below). In the next section, we show how to specialize this
protocol to the special case of many-to-one matchings, which gives
coordination complexity nearly matching our lower bound.

\begin{definition}
\label{def:convexprogram}
A linearly separable convex optimization problem consists of $n$
players and for each player $i$,
\begin{itemize}
\item a compact and bounded convex feasible set $\cF\ui\subseteq \{x\ui \in \RR^l \mid  \|x\ui\|\leq 1\}$,
\item a concave objective and $1$-Lipschitz function $v\ui \colon \cF\ui\rightarrow \RR$ such that $v\ui(\mathbf{0}) = 0$,
\item and $k$ convex constraint function $c_j\ui \colon \cF\ui \rightarrow [0,1]$ (indexed by $j = 1, \ldots , k$).
\end{itemize}
The convex optimization problem is:
\begin{align*}
\max_{x} &\sum_{i=1}^n v\ui(x\ui)\\
\mbox{subject to }& \sum_{i=1}^n c_j\ui(x\ui) \leq b_j \quad \mbox{for } j= 1, \ldots , k \iffull\quad \mbox{ (Coupling constraints)}\fi\\
& x\ui \in \cF\ui \quad \mbox{for } i = 1, \ldots, n \iffull\quad \mbox{ (Personal constraints)}\fi
\end{align*}
where each player $i$ controls the block of decision variable $x\ui$.
\end{definition}

Viewed as a coordination problem, the data held by each agent $i$ is
$D^{(i)} = \{\cF\ui\, v\ui, c_1\ui,\ldots,c_k\ui\}$, his action range
is $\cA_i = \cF\ui$, and the social objective function is $S$ the
objective of the convex program. \iffull\else We will call the first
set of constraints the~\emph{coupling constraints}, and the second set
of constraints the~\emph{personal constraints}.\fi

We will denote the product of the personal constraints by $\cF = \cF^{(1)}\times\ldots\times \cF^{(n)}$,
 the objective function by $v(x)$, and the optimal value by
 $\OPT$. In this notation we can write the problem as
\[
\max_{x\in\cF\ \textrm{and}\ \sum_{i=1}^n c_j\ui(x\ui) \leq b_j\ \textrm{for all } j} v(x).
\]
Note that here the problem is constrained both by the personal constraints $\cF$ and by the coupling constraints.
We will assume the problem above is feasible and our goal is
coordinate the players to play an aggregate solution $x =
(x\ui)_{i \in [n]}$ that is approximately feasible and optimal. \ar{Hm... What is the objective function, as we define the problem? I'm not sure how approximately feasible \emph{and} optimal fits into our framework.} Our
solution consists of two steps:
\begin{enumerate}
\item We will first introduce a regularization term $\eta \|x\|^2$ to our
 objective function, and coordinate the players to maximize the
 regularized objective. The purpose of adding this regularization term is to make the objective function \emph{strongly concave}, which will cause it to have the property that an optimal dual solution will uniquely specify the optimal primal solution.
\item Then we will show that the resulting optimal solution to the regularized problem is close to being optimal for the original (unregularized) problem. The weight of the regularization has to be traded off against the bit precision to which we need to communicate the optimal dual variables.
\end{enumerate}

\subsection{Coordination through Regularization}
In the first step, we add a small regularization term to our original
objective function. Consider the following convex optimization
problem:
\begin{align*}
\max_{x\in\cF\ \textrm{and}\ \sum_{i=1}^n c_j\ui(x\ui) \leq b_j\ \textrm{for all}\ j} &\; v'(x) = \sum_{i=1}^n v\ui(x\ui) - \frac{\eta}{2} \|x\|^2
\end{align*}
\begin{claim}
The objective function $v'$ is $\eta$-strongly concave.
\end{claim}

To solve the convex program, we will work with the partial
\emph{Lagrangian} $\cL(x, \lambda)$, which results from bringing only
the coupling constraints into the objective via Lagrangian dual
variables, but leaving the personal constraints to continue to
constrain the primal feasible region: \iffull
\begin{align*}
\cL(x, \lambda) &= \sum_{i=1}^n \left(v\ui\left( x\ui\right) - \frac{\eta}{2} \|x\ui\|^2 -
\sum_{j=1}^k \lambda_j\left( \sum_{i=1} c_j\ui\left( x\ui \right) - b_j \right)\right)\\
&= \sum_{i=1}^n \left(\left[ v\ui\left(x\ui\right) - \frac{\eta}{2} \|x_i\|^2 \right]-
\sum_{j=1}^k\lambda_j\left( \sum_{i=1} c_j\ui\left( x\ui \right) - b_j \right)\right)
\end{align*}
\else
\begin{align*}
  \sum_{i=1}^n \left(v\ui\left( x\ui\right) - \frac{\eta}{2} \|x\ui\|^2 -
\sum_{j=1}^k \lambda_j\left( \sum_{i=1} c_j\ui\left( x\ui \right) - b_j \right)\right)
\end{align*}
\fi

Let $\OPT'$ denote the optimum of the convex program, and by \emph{strong
duality} we have
\[
\max_{x\in \cF} \min_{\lambda\in \RR_+^k} \cL(x, \lambda) =
\min_{\lambda\in \RR_+^k} \max_{x\in \cF} \cL(x, \lambda) = \OPT'
\]
Fixing the optimal dual variables, $\lambda$, the optimal primal solution $y$ satisfies
\[
y = \argmax_{x\in \cF} \cL(x, \lambda)
\]
Note that the result of moving the coupling constraints into the Lagrangian is that we can now write the primal optimization problem over a feasible region defined only by the personal constraints. Because of this fact, and because the Lagrangian objective is linearly separable across players, given $\lambda$, each player's portion of the solution $y\ui$ is
\begin{align}\label{br}
\argmax_{x\in \cF_i} \left[v\ui\left( x\ui\right) -
  \frac{\eta}{2} \|x\ui\|^2 - \sum_{j=1}^k \lambda_j c_j\ui\left(x\ui\right)\right].
\end{align}

Thus, if the $\argmax$ were unique, this means that the optimal dual variables $\lambda$ would be sufficient to coordinate each of the parties to find their portion of the optimal solution, without the need for further communication (the problem, in general, is that the $\argmax$ need not be unique, and ties may need to be broken in a coordinated fashion). However, because we have added a strongly concave regularizer, the $\argmax$ is unique in our setting:

\begin{claim}
The solution to $$\argmax_{x\in \cF_i}\;v\ui\left( x\ui\right)
- \frac{\eta}{2} \|x\ui\|^2 - \sum_{j=1}^k \lambda_j\left( \sum_{i=1} c_j\ui\left( x\ui \right) \right)$$ is unique.
\end{claim}
\iffull
\begin{proof}
This follows from the fact that the function $v\ui\left( x\ui\right)
- \frac{\eta}{2} \|x\ui\|^2$ is strongly concave.
\end{proof}
\fi

This gives rise to our simple coordination mechanism $\reg$. The
mechanism first computes the optimal dual variables in our regularized partial Lagrangian problem, rounds them to finite precision, and then publishes these variables.  Then
each individual player finds her part of the near optimal solution by performing the optimization in \Cref{br}. The details are in
Algorithm~\ref{regc}.

\begin{algorithm}[H]
\label{regc}
\caption{Coordination Protocol for Linearly Separable Convex Programs $\reg(\eta, \eps)$}
  \begin{algorithmic}
    \STATE{\textbf{Input:} a linearly separable convex program
      instance $I$, regularization parameter $\eta$ and target
      accuracy $\eps$}
    \INDSTATE{Initialize: $\alpha = \frac{\eta\eps^2}{4\sqrt{nk}}$}
    \INDSTATE{Modify the objective of $I$ into
\[
\max_{x\in\cF} v(x) - \frac{\eta \|x\|^2}{2}
\]
}
    \INDSTATE{Compute the optimal dual solution $\lambda\ub$ for the modified convex program}
    \INDSTATE{Round each coordinate of $\lambda\ub$ into a multiple of $\alpha/\sqrt{k}$ and obtain $\hat\lambda$}
\iffull
    \INDSTATE{Broadcast the rounded dual solution $\hat\lambda$. To decode, each player $i$ computes:
      \[
\hat x\ui(\lambda) = \argmax_{x\in \cF_i} \left[v\ui\left( x\ui\right) -
  \frac{\eta}{2} \|x\ui\|^2 - \sum_{j=1}^k \hat\lambda_j c_j\ui\left(x\ui\right)\right]
      \]
\else
    \INDSTATE{Broadcast the rounded dual solution $\hat\lambda$. To decode, each player $i$ computes $\hat x\ui(\lambda)$:
      \[
 \argmax_{x\in \cF_i} \left[v\ui\left( x\ui\right) -
  \frac{\eta}{2} \|x\ui\|^2 - \sum_{j=1}^k \hat\lambda_j c_j\ui\left(x\ui\right)\right]
      \]
\fi
}
    \end{algorithmic}
\end{algorithm}

Next we show that the resulting solution is close to the optimal solution of the regularized convex
program (i.e. that we do not lose much by truncating the dual variables to have finite bit precision). Let $(x\ub, \lambda\ub)$ be an optimal primal-dual pair for
the regularized convex program. Note that since the objective of the
program is strongly concave, $x\ub$ is unique. First, we will
show that if the broadcast dual vector $\hat \lambda$ is close to an
optimal dual solution $\lambda\ub$, the resulting solution $\hat x$
will also be close to the optimal primal solution $x\ub$.

\begin{restatable}{lemma}{dualprimal}
\label{dual-primal}
Suppose we have a dual vector
$\hat \lambda$ such that $\|\lambda\ub - \hat\lambda\| \leq \alpha$.
Let $\hat x = \argmax_{x\in \cF} \cL(x, \hat \lambda)$, then
\[
 \|\hat x - x\ub\| \leq \frac{2 \sqrt{\alpha} (nk)^{1/4}}{\sqrt \eta}
\]
\end{restatable}

The proof relies on some basic properties of the Lagrangian and strong
concavity and is deferred to the \iffull appendix\else full version\fi.

\begin{lemma}\label{noconstants}
The coordination mechanism $\reg$ instantiated with regularization
parameter $\eta$ and target accuracy parameter $\eps$ will coordinate
the players to play a solution $\hat x$ that satisfies $\|\hat x -
x\ub\|\leq \eps$, and has a coordination complexity of
$O(k\log(nk/\eta\eps))$.
\end{lemma}

\begin{proof}
Note that $\alpha = \frac{\eta \eps^2}{4 \sqrt{nk}}$, and the
mechanism rounds each coordinate of the optimal dual solution
$\lambda\ub$ to a multiple of $\alpha/\sqrt{k}$, so the approximate
dual vector $\hat \lambda$ can be specified with $O(k \log(\sqrt{k}
/ \alpha))$ bits.

Since for each coordinate $j$, $|\lambda_j\ub - \hat \lambda_j| \leq
\alpha/\sqrt k$, we also have that that $\|\lambda\ub -
\hat\lambda\| \leq \alpha$. By~\Cref{dual-primal}, we know that
$\|\hat x - x\ub\| \leq \eps$.
\end{proof}

\subsection{Approximate Feasibility and Optimality}
Now we carry out the second step to show that if we choose the
regularization parameter $\eta$ carefully, the solution resulting from
the coordination mechanism above is both approximately feasible and
optimal. Let $x^*$ denote the optimal solution of the original convex
program, $x\ub$ denote the optimal solution of the regularized convex
program, and $\hat x(\eta)$ denote the solution resulting from the
coordination mechanism when we use parameter $\eta$.

As an intermediate step, we will first bound the objective difference
between $x\ub$ and $x^*$.

\begin{lemma}
For any choice of $\eta$, $v(x^*) - v(x\ub) \leq \frac{\eta n}{2}$.
\end{lemma}
\iffull
\begin{proof}
Since both $x^*$ and $x\ub$ are in the feasible region of the
regularized convex program, we know that
\[
v(x\ub) - \frac{\eta}{2} \|x\ub\|^2 \geq v(x^*) - \frac{\eta}{2}\|x^*\|^2.
\]
Since for each $i$, $(x\ub)\ui, (x^*)\ui$ has $\ell_2$ norm bounded by
1, then $\|x\ub\|^2 - \|x^*\|^2 \leq n$. Therefore, we must have
$v(x\ub) \geq v(x^*) - \frac{\eta n}{2}$
\end{proof}
\fi

Next we bound the objective difference between $\hat x$ and $x\ub$
using Lipschitzness.

\begin{lemma}
  Suppose that $\|\hat x - x\ub \| \leq \eps$, then $$v(x\ub) - v(\hat
  x)\leq n\eps.$$
\end{lemma}
\iffull
\begin{proof}
The proof follows easily from the fact that each $v\ui$ is 1-Lipschitz
and the function $v$ is $n$-Lipschitz in the aggregate vector $x$.
\end{proof}
\fi

\begin{theorem}
The coordination mechanism $\reg(\eta, \eps)$ coordinates the players
to play a joint solution $\hat x$ that satisfies
\[
v(\hat x) \geq \OPT - n(\eps + \eta) \qquad \mbox{and} \qquad \min_{x\in \cF} \|x- \hat x\| \leq \eps
\]
and has coordination complexity of $O(k \log(nk / \eta\eps))$.
\end{theorem}

\begin{proof}
Follows easily from the previous lemmas.
\end{proof}

\subsection{Application to Many-to-One Matchings}
Next we show a simple instantiation of our coordination mechanism for
linearly separable convex programs to give a coordination complexity
upper bound for many-to-one matchings.  First, let's consider the
following linear program formulation of the matching problem.\sw{I
  decided to write unweighted matching, which is more consistent with
  lower bound and gives multiplicative approx}\ar{Did you switch back?
  Doesn't look unweighted to me.}
\begin{align}
\max_{x} &\sum_{i=1}^n \sum_{j = 1}^k  \; v_{i,j} \, x_{i,j} \label{eq:obj}\\
\mbox{subject to }& \sum_{i=1}^n x_{i,j} \leq b_j \quad \mbox{for } j= 1, \ldots , k \label{eq:supply}\\
& \sum_{j=1}^k x_{i,j} \leq 1 \quad \mbox{for } j= 1, \ldots , k\label{eq:1dude}\\
& x_{i,j} \geq 0 \quad \mbox{for } i = 1, \ldots, n \mbox{ and } j = 1, \ldots , k\label{eq:nnegative}
\end{align}

Observe that the matching linear program is an example of
a linearly separable convex program
as defined in~\Cref{def:convexprogram}. Each player $i$ has valuation $v_{i,j}\in
\{0,1\}$ for each type of good $j$ and controls the decision variables
$\{x_{i,j}\}_{j=1}^k$. Each supply constraint in~\Cref{eq:supply}
corresponds to a coupling constraint, and constraints in
both~\Cref{eq:1dude} and~\Cref{eq:nnegative} are personal constraints.

A nice property about the matching linear program is that any extreme
point is integral.\sw{cite!} However, this structure no longer holds
if we add a regularization term to the welfare objective, so the
resulting solution $\hat x$ resulting from the coordination mechanism
will be fractional.  To obtain an integral solution, we can simply use
independent rounding, which does not require any further
coordination. In order to obtain an integral solution, each player $i$
will take their portion of the fractional solution $(\hat
x_{i,j})_{j=1}^k$ and will independently sample a good by selecting
each good $j$ with probability $\hat x_{i,j}$. We will continue to use
similar notation: let $v(\cdot)$ denote the welfare objective in the
linear program, let $x^*$ be the optimal solution for the matching
linear program with welfare $\OPT$, $\hat x$ be the optimal solution
for the regularized program with welfare $\hat V$, $x'$ be the rounded
solution of $\hat x$, and let $\cF$ denote the feasible region defined
by all the constraints of~\Cref{eq:supply} in the linear program. The
following lemma bounds the loss of welfare due to rounding. \iffull For
details of the proof, see the full version.\fi

\begin{lemma}
Let $\beta\in (0,1)$. Then with probability at least $1-\beta$, the
rounded solution $x'$ satisfies
\[
v(x') \geq \left(1 - \frac{\log(2/\beta)}{\sqrt{\hat V}} \right) \hat V
\]
\end{lemma}
\iffull
\begin{proof}
Since
$\Expectation\left[\sum_{i=1}^n\sum_{j=1}^k x'_{i,j} \right] = \hat V$, by
Chernoff-Hoeffding bound, we know that for any $\delta \in (0,1)$,
\[
\Pr\left[ \sum_{i=1}^n\sum_{j=1}^k  x'_{i,j} < (1 - \delta) \hat V\right] < \exp\left( -\delta^2 \hat V /2\right)
\]
If we set $\beta = \exp(-\delta^2\hat V / 2)$, then we get $\delta
= \sqrt{2\log(1/\beta) / \hat V}$, which recovers the stated bound.
\end{proof}
\fi
Now we look at approximate feasibility of $x'$.

\begin{restatable}{lemma}{matchingviolation}
  Suppose that $\min_{x\in \cF}\| x - \hat x \| \leq \eps$. Then with
  probability $1 - \beta$, $x'$ satisfies
\[
\sum_{j=1}^k \left(\sum_{i=1}^n x'_{i,j} - b_j \right)_+ \leq \sqrt{3k\log(k/\beta)\hat V} + \sqrt{nk}\eps
\]
\end{restatable}

Observe that since this is a packing linear program, if desired, it is easy to obtain exact feasibility by simply scaling down the supply constraints: this transfers the approximation factor in the feasibility bound to the become an approximation factor in the objective.


Lastly, we are ready to establish the welfare guarantee for the
rounded solution. Since the solution we obtain might slightly violate the feasibility constraints, we want to make sure we are not over-counting. If more than $b_j$ parties select a particular good of type $j$, we only count the first $b_j$ parties to select it when measuring our welfare guarantee.

\begin{restatable}{theorem}{mainconvex}\label{mainconvex}
There exists a coordination protocol with coordination complexity
$O(k\log(nk))$ such that the parties coordinate on a matching $x'$
with total weight:
  \[
  \sum_{j=1}^k \min\left\{\sum_{i=1}^n v_{i,j} x'_{i,j}, b_j\right\}
  \geq \left(1 - O\left(\frac{\sqrt{k}\log(k/\beta)}{\sqrt{\OPT} }
  \right) \right)\OPT
  \]
as long as $\OPT \geq 1$.
\end{restatable}

Observe that in the setting of many-to-one matchings, when the supply of each good is $s_j \gg 1$, we expect that $\OPT \gg k$, and hence in this setting, the above theorem guarantees a solution with weight $(1-o(1))\OPT$.

\iffull
\begin{remark}
We remark that many other combinatorial optimization tasks have fractional relaxations that can be written as linearly separable convex programs, and the same rounding technique can be applied to get low-coordination-complexity protocols for them. This class includes among others multi-commodity flow (where the coordination complexity scales with the number of edges in the underlying graph, but not with the number of parties who wish to route flow) and multi-dimensional knapsack problems (where the coordination complexity scales with the number of different types of knapsack constraints, but not with the number of parties who need to decide on their inclusion in or out of the knapsack).
\end{remark}
\fi

\section{Interface with Privacy and Efficiency in Games}\label{s.privacy}

In this section, we explain a simple implication of our results:
Problems that have low sensitivity objectives (i.e. problems such that
one party's data and action do not substantially affect the objective
value) and low coordination complexity also have good algorithms for
solving them subject to \emph{joint differential privacy}. When the
strategic variant of the optimization problem is a smooth game, they
also have good welfare properties for no-regret players, even when
agent types are dynamically changing.

\subsection{Privacy Background}
A \emph{database} $D\in \cX^n$ is an $n$-tuple of private records,
each from one of $n$ agents. Two databases $D, D'$ are
\emph{$i$-neighbors} if they differ only in their $i$-th index: that
is, if $D_j = D_j'$ for all $j\neq i$. If two databases $D$ and $D'$
are $i$-neighbors for some $i$, we say that they are \emph{neighboring
  databases}. We write $D\sim D'$ to denote that $D$ and $D'$ are
neighboring. We will be interested in randomized algorithms that take
a database as input, and output an element from some abstract range
$\cR$.

\begin{definition}[\cite{DMNS06}]
A mechanism $\cM\colon \cX^n \rightarrow \cR$ is $(\eps,
\delta)$-\emph{differentially private} if for every pair of
neighboring databases $D, D'\in \cX^n$ and for every subset of outputs
$\cS\subseteq \cR$,
\[
\Pr[\cM(D) \in \cS] \leq \exp(\eps) \Pr[\cM(D') \in \cS] + \delta.
\]
\end{definition}

For the class of problems we consider, elements in both the domain and the
range of the mechanism are partitioned into $n$ components, one for each
player. In this setting, \emph{joint differential privacy}~\citep{KPRU14} is a
more natural constraint: For all $i$, the \emph{joint} distribution on
outputs given to players $j\neq i$ is differentially private in the input of
player $i$. Given a vector $x = (x_1, \ldots, x_n)$, we write $x_{-i} = (x_1,
\ldots, x_{i-1}, x_{i+1} , \ldots, x_n)$ to denote the vector of length
$(n-1)$ which contains all coordinates of $x$ except the $i$-th coordinate.

\begin{definition}[\cite{KPRU14}]
A mechanism $\cM \colon \cX^n \rightarrow \cR^n$ is $(\eps,
\delta)$-\emph{jointly differentially private} if for every $i$, for
every pair of $i$-neighbors $D, D'\in \cX^n$, and for every subset of
outputs $\cS\subseteq \cR^{n-1}$,
\[
\Pr[\cM(D)_{-i} \in \cS] \leq \exp(\eps) \Pr[\cM(D')_{-i} \in \cS] + \delta.
\]
If $\delta = 0$, we say that $\cM$ is $\epsilon$-differentially private. The case of $\delta > 0$ is sometimes referred to as \emph{approximate} differential privacy.
\end{definition}
Note that this is still a very strong privacy guarantee; the mechanism preserves
the privacy of any player $i$ against arbitrary coalitions of other players. It
only weakens the constraint of differential privacy by allowing player $i$'s
output to depend arbitrarily on her \emph{own} input.

An important class of jointly differentially private algorithms -- particularly amenable to our purposes -- are those that work in the so-called  \emph{billboard model}. Algorithms in the billboard
model compute a differentially private signal as a function of the input database; then each player $i$'s portion of the output is computed as a
function only of this private signal and the private data of player
$i$. The following lemma shows that algorithms operating in the
billboard model satisfy joint differential privacy.

\begin{lemma}[\cite{HHRRW14}] \label{billboard}
  Suppose $\cM : \cX^n \rightarrow \cR$ is $(\eps,
  \delta)$-differentially private. Consider any set of functions $f_i
  : \cX \times \cR \rightarrow \cR'$. Then the mechanism $\cM'$ that
  outputs to each player $i$: $f_i(D_i, \cM(D))$ is $(\eps,
  \delta)$-jointly differentially private.
\end{lemma}

Note the similarity between algorithms operating in the billboard model and coordination complexity protocols: a signal is computed by a central party, and then the action of each agent is a function only of this signal and of their own portion of the problem instance. Thus, the following lemma is immediate:

\begin{lemma}\label{privatepro}
  A coordination protocol $(\sigma, \pi)$ satisfies $(\eps,
  \delta)$-joint differential privacy if the coordinator's encoding
  function $\sigma$ satisfies $(\eps, \delta)$-differential privacy.
\end{lemma}
\iffull
\begin{proof}
Follows from~\Cref{billboard}.
\end{proof}
\fi
\subsection{A Generic Private Coordination Protocol}
Next, we give a general way to convert any coordination protocol to a
jointly differentially private algorithm -- and the lower the
coordination complexity of the protocol, the better the utility
guarantee of the private algorithm. The tool we use is the
\emph{exponential mechanism} of~\cite{MT07}, one of the most basic
tools in differential privacy. To formally define this mechanism, we
consider some arbitrary range $\cR$ and some quality score function $q
\colon \cX^n \times \cR \rightarrow \RR$, which maps database-output
pairs to quality scores.

\begin{definition}[The Exponential Mechanism~\cite{MT07}]
  The \emph{exponential mechanism} $\cM_E(D, q, \cR, \eps)$ selects
  and outputs an element $r\in \cR$ with probability proportional to
  $$\exp\left(\frac{\eps q(D, r)}{2\Delta(q)} \right),$$ where
  \[
  \Delta(q) \equiv \max_{D, D'\in \cX^n, D\sim D'} |q(D) - q(D')|.
  \]
\end{definition}

McSherry and Talwar showed that the exponential mechanism is private
and with high probability selects an outcome with high quality.

\begin{theorem}[\cite{MT07}]
  The exponential mechanism \iffull$\cM_E(\cdot, q, \cR, \eps)$\fi satisfies
  $(\eps, 0)$-differential privacy, and for any $D\in \cX^n$ it
  outputs an outcome $r\in \cR$ that satisfies
  \[
  q(D, r) \geq \max_{r'} q(D, r') - \frac{2\Delta(q) (\log(|\cR|/\beta))}{\eps}
  \]
with probability at least $1 - \beta$.
\end{theorem}

Using the exponential mechanism, we can take any coordination protocol
$(\sigma, \pi)$, and construct a jointly differentially private
coordination protocol $(\sigma', \pi)$ with the same coordination
complexity, and almost the same approximation factor.  The idea is to
construct a differentially private encoding function $\sigma'$ that
selects from the message space of $\sigma$ using the exponential
mechanism.  Without loss of generality, we assume that the social objective
function $S$ has low-sensitivity:
\iffull
\[
\max_{i\in [n]}\max_{a\in \cA^n, D\in \cX^n} \max_{a_i'\in \cA_i, D_i'\in \cX} |S(D, a) - S((D_i',
D_{-i}), (a'_i, a_{-i}))| \leq 1.\footnote{We can always obtain
  this condition by scaling. It is already satisfied in the matching problem.}
\]
\else
\[
\max_{i\in [n], a\in \cA^n, D\sim D', a_i'\in \cA_i} |S(D, a) - S(D', (a'_i, a_{-i}))| \leq 1.\footnote{We can always obtain
  this condition by scaling. It is already satisfied in the matching problem.}
\]
\fi

Now we present the private protocol as follows:

\begin{algorithm}[H]
\label{alg:privatize}
\caption{Jointly private algorithm $\pritocol((\sigma, \pi), q, \eps, D)$}
  \begin{algorithmic}
    \STATE{\textbf{Input:} A coordination protocol $(\sigma, \pi)$, objective function $f$, and input instance $D$}
    \INDSTATE{Let $\cR = \{\sigma(D') \mid D'\in \cX^n\}$ be the space of all possible messages in the range of $\sigma$}
    \INDSTATE{Let quality function $q$ be defined as $q(D, r) = \Expectation_{\pi}[S(D, (\pi(r, D\ui) )_{i\in [n]})]$ $\forall D\in \cX^n, r\in \cR$}

    \INDSTATE{Let $\sigma'(D) = \cM_E(D, q, \cR)$ be the message selected by the exponential mechanism}
    \STATE{\bf Output $a = (\pi(\sigma'(D), D\ui))_{i=1}^n$}
    \end{algorithmic}
\end{algorithm}

\begin{lemma}\label{coorprivate}
Suppose that $(\sigma,\pi)$ has coordination complexity
$\ell$ and approximation ratio $\rho$ for the objective $f$. Then the algorithm $\pritocol((\sigma,
\pi), f, \eps, D)$ satisfies $(\eps, 0)$-joint differential privacy,
and with probability at least $1 - \beta$, the resulting action
profile $a$ satisfies
\[
\Expectation\left[S(D, a)\right]  \geq \frac{\OPT(D)}{\rho} - \frac{2 (\ell + \log(1/\beta))}{\eps},
\]
where the expectation is taken over the internal randomness of the
encoding function $\sigma'$ and decoding function $\pi$.
\end{lemma}
\iffull
\begin{proof}
Since the encoding function is an instantiation of the exponential
mechanism, we know from~\Cref{privatepro} that the instantiation
$\pritocol((\sigma, \pi), q, \eps, D)$ satisfies $(\eps, 0)$-joint
differential privacy.

  Since the coordination protocol guarantees an approximation ratio of
  $\rho$, there exists some message $r\ub$ in the set $\cR$ such that
  \[
  \Expectation_{\pi}\left[S(D, (\pi(r\ub , D\ui))_{i\in [n]})\right] = q'(D, r\ub) \geq \frac{\OPT(D)}{\rho}.
  \]
Note that $\Delta(q') \leq 1$ by our assumption on $f$. Then the
utility guarantee of the exponential mechanism gives
  \[
  \Expectation_{\pi}\left[S(D, (\pi(\sigma'(r, D) , D\ui)))_{i\in [n]})\right] = q'(D, \pi(\sigma'(r, D))) \geq \max_{r\in \cR}q'(D, r) - \frac{2(\log(|\cR| - \log(1/\beta)))}{\eps}
  \]
Also, observe that $\max_{r\in \cR}q'(D, r) \geq \OPT(D)/\rho$ and
$\log(|\cR|) \leq \ell$, so we also have
\[
\Expectation_{\sigma}\left[S(D, (\pi(\sigma'(r, D) , D\ui)))_{i\in
    [n]})\right] \geq \OPT(D)/\rho - \frac{2(\ell - \log(1/\beta)))}{\eps}
\]
which recovers our stated bound.
\end{proof}
\fi

\subsection{Efficiency in Games with Dynamic Population}
Now we briefly discuss a connection between coordination complexity
and the efficiency in games with dynamic population, which leverages
the connection to joint differential privacy discovered
by~\cite{dynamicPOA}.  We briefly introduce the model
in~\cite{dynamicPOA}, but the discussion will necessarily be lacking
in detail -- see~\cite{dynamicPOA} for a formal treatment.

Let $G$ be an $n$-player normal form \emph{stage game}. We consider
this game played repeatedly with a changing population of players over
$T$ rounds. Each player $i$ has an action set $\cA_i$, type $D\ui$,
and a utility function $u(D\ui, a) = u_i(a)$. For concreteness, we can
think about allocation games defined by auction rules $M$, which take
as input an action profile and output an allocation $X_i(a)$ and a
payment $P_i(a)$ for each player. Players have quasi-linear utility
$u_i(a) = v(D\ui, X_i(a)) - P_i(a) = v_i(X_i(a)) - P_i(a)$, where
$v_i\colon \cA^n\rightarrow[0,1]$ denotes the valuation of player $i$
over the allocation. In these games, a natural objective function is
social welfare: $S(D, a) = \sum_{i=1}^n v_i(X_i(a))$. We write
$\OPT(D) = \max_{a\in \cA^n} S(D, a)$ to denote the optimal welfare
with respect to an instance $D$.

In the model of~\cite{dynamicPOA}, after each
round, every player independently exits with some probability $p$.
Whenever a player leaves the game, she is replaced a new player, whose type is chosen adversarially.  We
will write $D^t$ to denote the game instance, and $a^t$ to denote the
action profile played at round $t$. Lastly, we also assume that each player
in the game is a \emph{no-regret learner} and plays some
\emph{adaptive learning algorithm}.\footnote{For more details of
  adaptive learning algorithms and adaptive regret, see~\cite{HS07}.}

The main result of~\cite{dynamicPOA} is that the existence of
jointly differentially private algorithms that find action profiles approximately
optimizing the welfare in a game implies that when the dynamically changing game is played by no-regret players, their average welfare is high.
\iffull
\begin{theorem}[Corollary 5.2 of~\cite{dynamicPOA}]\label{dPoA}
  Consider a mechanism with dynamic population $(M, T, p)$, such that
  the stage mechanism $M$ is allocation based $(\lambda, \mu)$-smooth
  and $T\geq 1/p$. Assume that there exists an $(\eps, \delta)$-joint
  differentially private allocation algorithm $X\ub\colon \cX^n
  \rightarrow \cA^n$ such that for any input instance $D\in \cX^n$ it
  computes a feasible outcome that is $\rho$-approximately optimal
  \[
  \Expectation[S(D, X\ub(D)] \geq \OPT(D)/\rho.
  \]

If all players use adaptive learning in the repeated mechanism, then
the overall welfare satisfies
  \[
  \sum_t \Expectation[S(D^t, a^t)] \geq \frac{\lambda}{\rho\max\{1, \mu\}}
  \sum_t \Expectation[\OPT(D^t)] -
  \frac{nT}{\max\{1, \mu\}}  \sqrt{2p(1 + n(\eps + \delta)) \ln(NT)},
  \]
where $N = \max_i |\cA_i|$.
\end{theorem}
\else
\begin{theorem}[Corollary 5.2 of~\cite{dynamicPOA}]\label{dPoA}
  Consider a mechanism with dynamic population $(M, T, p)$, such that
  the stage mechanism $M$ is allocation based $(\lambda, \mu)$-smooth
  and $T\geq 1/p$. Assume that there exists an $(\eps, \delta)$-joint
  differentially private allocation algorithm $X\ub\colon \cX^n
  \rightarrow \cA^n$ such that for any input instance $D\in \cX^n$ it
  computes a feasible outcome that is $\rho$-approximately optimal
  \[
  \Expectation[S(D, X\ub(D)] \geq \OPT(D)/\rho.
  \]

If all players use adaptive learning in the repeated mechanism, then
the overall welfare satisfies
  \begin{align*}
  \sum_t \Expectation[S(D^t, a^t)] \geq \frac{\lambda}{\rho\max\{1, \mu\}}
  \sum_t \Expectation[\OPT(D^t)] - \\
  \frac{nT}{\max\{1, \mu\}}  \sqrt{2p(1 + n(\eps + \delta)) \ln(NT)},
  \end{align*}
where $N = \max_i |\cA_i|$.
\end{theorem}
\fi

\iffull Note that if the problem of coordinating a high-welfare
allocation has small coordination complexity this implies the
existence of a jointly differentially private allocation algorithm
with a high welfare guarantee. By combining~\Cref{dPoA}
and~\Cref{coorprivate}, we obtain the following result.  \else We show
the following result connecting low coordination complexity with the
welfare guarantee in games with dynamic population. (See the full
version for details).
\fi

\begin{lemma}
Consider a mechanism with dynamic population $(M, T, p)$ such that the
stage mechanism $M$ is allocation based $(\lambda, \mu)$-smooth and
$T\geq 1/p$.  Assume there is a coordination protocol $(\sigma, \pi)$
with coordination complexity $\ell$ and approximation ratio $\rho$
for the corresponding welfare maximization problem.

Then if all players use adaptive learning in the repeated mechanism,
the average welfare satisfies
\iffull
\begin{align*}
&  \sum_t \Expectation\left[ S(D^t, a^t) \right] \geq 
  \frac{\lambda}{\rho \max\{1, \mu\}} \sum_t \Expectation[\OPT(D^t)]\; -\;\\
 &\inf_{\eps > 0} \left\{\frac{nT}{\max\{1, \mu\}} \sqrt{4p n\eps \ln(NT)} +
   \frac{\lambda T}{\rho\max\{1, \mu\}} \frac{2(\ell + \log(n))}{\eps}\right\}.
\end{align*}
\else
\begin{align*}
&  \sum_t \Expectation\left[ S(D^t, a^t) \right] \geq 
  \frac{\lambda}{\rho \max\{1, \mu\}} \sum_t \Expectation[\OPT(D^t)] \\
 &- \frac{T}{\max\{1, \mu\}}\cdot \inf_{\eps > 0} \left\{{n} \sqrt{4p n\eps \ln(NT)} +
   \frac{\lambda}{\rho} \frac{2(\ell + \log(n))}{\eps}\right\}.
\end{align*}

\fi
\end{lemma}

\section{Coordination through Dynamics Simulation}
In this section, we give another general technique for designing
coordination protocols. The key idea is to broadcast a message that is
sufficient for players to derive the sequence of actions that they
would have executed in some joint dynamic that is known to converge to
the solution to the coordination problem. Similar techniques have been
used in the privacy literature, for example~\cite{KMRW15},
\cite{HHRRW14} and \cite{RR14}. We will focus on the application of
coordinating an equilibrium flow in atomic routing games, which follows the general outline of \cite{RR14}.

A basic primitive that turns out to be broadly useful when writing down the transcript of some dynamic is being able to keep a running count of a stream of numbers. For many applications, in fact, it is sufficient to be able to maintain an \emph{approximate} count. Before we start, we
introduce two subroutines for keeping track of the approximate count of a binary stream using low communication. The first one compresses a numeric  stream
$\tau\colon [T] \rightarrow \{-1,0,1\}$ into a short
message, and the second one decompresses it. See~\Cref{compress} for the simple compression protocol.

\begin{algorithm}[h]
  \caption{$\approxC(\tau, r, T)$}
 \label{compress}
  \begin{algorithmic}
    \STATE{\textbf{Input:} a stream of numbers $\tau
      \colon [T] \rightarrow \{-1,0,1\}$ and refinement parameter
      $r$}
    \INDSTATE[1]{\textbf{Initialize:} a counter $C\colon [T] \rightarrow \NN$, a list of update steps $U = \emptyset$}
    \INDSTATE[1]{\textbf{for } $t = 1, \ldots , T$:}
    \INDSTATE[2]{\textbf{if} $t=1$: let $C(t) = 0$; {\bf else}: let $C(t) = C(t-1)$}
    \INDSTATE[2]{\textbf{if } $|C(t) - \sum_{i=1}^t \tau(i)| \geq  r$: }

    \INDSTATE[3]{\textbf{if } $C(t) <  \sum_{i=1}^t \tau(i)$ \textbf{then} $C(t) = C(t - 1) + r$ and $U \leftarrow U\cup \{(t, +)\}$}
    \INDSTATE[3]{\textbf{if } $C(t) >  \sum_{i=1}^t \tau(i)$ \textbf{then} $C(t) = C(t-1) -r$ and $U \leftarrow U\cup \{(t, -)\}$}
    \STATE{\textbf{Output:} the list of time steps $U$}
    \end{algorithmic}
\end{algorithm}
This algorithm releases a concise description $U$ that suffices to reconstruct an approximate running count $C(t)$ of the stream.
\begin{claim}~\label{counteracc}
For all $t \in [T]$, the approximate count $C(t)$
satisfies $|C(t) - \sum_{i=1}^t \tau(i)| \leq r$. The summary statistic $U$ can be written with $O\left(\frac{\|\tau\|_0 \log T}{r} \right)$
bits.
\end{claim}

The second one takes the compressed message from $\approxC$ as input,
and extracts the approximate counts. See~\Cref{extract}.

\begin{algorithm}[h]
  \caption{$\exstream(U, r, T)$}
 \label{extract}
  \begin{algorithmic}
    \STATE{\textbf{Input:} a list of update steps $U$, refinement
      parameter $r$, and time horizon $T$}
    \INDSTATE[1]{\textbf{Initialize:} a counter $C \colon [T]
      \rightarrow \NN$ such that $\tau(t) = 0$ for all $t\in [T]$}
    \INDSTATE[1]{\textbf{for:} each $(t, \bullet)\in U$}
    \INDSTATE[2]{Let $c' = C(t - 1)$}
    \INDSTATE[2]{{\bf if} $\bullet = +$: $c = c' + r$; {\bf else}: $c = c' - r$}
    \INDSTATE[2]{\textbf{for:} each $t' \in \{t, \ldots, T\}$:}
    \INDSTATE[3]{Let $C(t') = c$}
    \STATE{\textbf{Output:} the approximate counts $C$}
    \end{algorithmic}
\end{algorithm}

Looking ahead, we will use $\approxC$ in the coordinator's encoding
function to compress the information in the simulated dynamics, and
use $\exstream$ in each player's decoding function to extract the
information about the dynamics.

\subsection{Atomic Routing Games and Best-Response Dynamics}
An \emph{atomic routing game} instance is defined by a directed graph
$G= (V, E)$, $n$ players with their source-sink pairs $(s_1, d_1),
\ldots, (s_n, d_n)$, and a continuous, nondecreasing and
$\lambda$-Lipschitz cost function $c_e \colon [0, n] \rightarrow
      [0,1]$ for each edge $e\in E$.  Each player $i$ needs to route 1
      unit of flow from $s_i$ to $d_i$, so her strategy set $\cA_i$ is
      the set of $s_i$-$d_i$ paths. We think of a flow of a single player alternately as a vector indexed by paths $P$, and as a vector indexed by edges $e$. The aggregate flow is the sum over all player flows. A flow $f$ (viewed as a vector indexed by paths) is \emph{feasible} if for each player $i$, $f^{(i)}_P$ equals 1
      for exactly one $s_i$-$d_i$ path and equals 0 for all other
      paths. We can translate such a flow into a flow indexed by edges by defining $f_e^{(i)} = \sum_{e \in P}f_P^{(i)}$. The cost $c_P(f)$ of path $P$ in a flow $f$ is
\[
c_P(f) = \sum_{e\in P} c_e(f_e)
\]
where $f_e = \sum_{i=1}^n f_e^{(i)}$. We will denote the number of
edges by $m$, the set of feasible flow by $\cF$, and sometimes abuse
notation to use $f\ui$ to denote the path of player $i$. Now we
formalize the notion of approximate equilibrium flow in a routing
game.

\begin{definition}[Approximate Equilibrium Flow]
Let $f$ be a feasible flow for the atomic instance $(G, c)$. The flow
$f$ is an $\eps$-\emph{equilibrium flow} if each player $i\in [n]$ is
playing $\eps$\emph{-best-response}, that is for every pair every pair of $s_i-d_i$  
paths $P, P'$ with $f_P^{(i)} > 0$,
\[
c_P(f) \leq c_{P'}(f') + \eps
\]
where $f'$ is the flow identical to $f$ except for its $i$-th
component: $f'^{(i)}_P = 0$ and $f'^{(i)}_{P'} = 1$. When $f$ is a
$0$-equilibrium flow, we simply say that $f$ is a equilibrium flow.
\end{definition}

The classical work of~\cite{MS96} establishes the existence of
equilibrium flows and shows that an equilibrium flow minimizes the
\emph{potential function} $\Psi$ of the routing game, which is defined
as
\[
\Psi(f) \equiv \sum_{e\in E}\sum_{i=1}^{f_e} c_e(i).
\]

Note that in atomic routing games, equilibrium flows are not unique, and so equilibrium selection is important. This motivates the coordination problem we study. 

We will rely on the following fact to show that a flow is an
approximate equilibrium.

\begin{fact}
Consider a flow $f\in \cF$. Suppose a player $i$ could decrease her
cost by deviating from path $P$ to path $\tilde P$, which gives rise
to a new flow $\tilde f$, then
\[
c_P(f) - c_{\tilde P}(\tilde f) = \Psi(f) - \Psi(\tilde f).
\]
\end{fact}

Our goal is to give a coordination protocol that coordinates the
players to play an approximate equilibrium flow and has low
coordination complexity (scaling with the number of edges $|E|$
instead of number of players $n$).\footnote{The social objective here
  is not social welfare --- we want instead to minimize approximation factor of
  the equilibrium $\eps$.} A very straightforward procedure to compute
an equilibrium flow is the \emph{best-response dynamics}: while the
flow $f$ is not a $\eta$-equilibrium flow, pick a player $i$ and an
arbitrary path deviation that could decrease her cost. In our
coordination protocol, the coordinator will first simulate an
approximate version of the best-response dynamics, and compress the
dynamics into a concise string using the subroutine $\approxC$; then
the coordinator will broadcast the string to the players so that each
player can simulate the sequence of actions she would have played in the dyamics, and thus 
determine the action she plays at the end of the dynamics using $\exstream$.

In our approximate best-response dynamics, we will let the players
best-respond to the approximate count of players on the edges. We
first need to define the a player's best response with respect to a
count vector in $\RR^m$.

\begin{definition}[Best-Response with respect to counts]
Given a count vector $f\ub\in \RR^{|E|}$, a path $P\ui$ for player $i$
is an $\eps$-best-response with respect to the vector $f\ub$ if
\[
c_{P\ui}(f\ub) - \eps \leq \min_{\tilde P\ui \in \cA_i} c_{\tilde P\ui}(f\ub).
\]
\end{definition}
Keep in mind that any feasible flow $f$ is also a count vector. We
give the formal description of the coordinator's encoding function
$\br$ in~\Cref{br-counter}.

\begin{algorithm}[h]
\caption{$\br((G, \{(s_i, d_i)\}_{i\in [n]}, c), \alpha, r)$}
 \label{br-counter}
  \begin{algorithmic}
    \STATE{\textbf{Input:} a routing game instance $(G, \{(s_i,
      d_i)\}_{i\in [n]}, c)$, best-response parameter $\alpha$ and
      refinement parameter $r$ such that $\alpha > 2\lambda m(r + 1)$ }
    \STATE{\textbf{Initialize:}
      \[
      l = \alpha -  2\lambda m r  - \lambda m, \qquad T = \frac{mn}{l}
      \]}
    \INDSTATE{\textbf{for } each edge $e\in E$}
    \INDSTATE[2]{Let $\counter(e)$ be an instantiation of $\approxC(\cdot, r, (T + 1) n )$ (waiting for incoming stream)}
    \STATE{\bf Form the initial flow:}
    \INDSTATE[1]{\textbf{for } each player $i$:}
    \INDSTATE[2]{Let $f\ui$ be the $s_i$-$d_i$ path $P\ui$ with the fewest number of edges, break ties lexicographically}
    \INDSTATE[3]{\textbf{for } each edge $e$:}
    \INDSTATE[4]{\textbf{if } $e\in P\ui$ send ``1'' to $\counter(e)$ \textbf{else } send ``0'' to $\counter(e)$}

    \STATE{\bf Best-responses dynamics:}
    \INDSTATE{\textbf{for} $t = 1, \ldots , T$}
    \INDSTATE[2]{{\bf if} each player is playing an $\alpha$-best-response w.r.t. the  counts $\{\counter(e)\}_{e\in E}$: Halt}
    \INDSTATE[2]{\textbf{for} each player $i$}
    \INDSTATE[3]{\textbf{if} $i$ is not playing an $\alpha$-best-response w.r.t. the flow $\{\counter(e)\}_{e\in E}$}
    \INDSTATE[4]{Let $\hat f\ui$ be the best-response of $i$ w.r.t. $(\counter(e))_{e\in E}$ (breaking ties lexicographically)}
    \INDSTATE[4]{{\bf for} each $e$:}
    \INDSTATE[5]{{\bf if } $e\in f\ui \setminus \hat f\ui$: Send ``-1'' to $\counter(e)$}
    \INDSTATE[5]{{\bf if } $e\in \hat f\ui \setminus f\ui$: Send ``1'' to $\counter(e)$}
    \INDSTATE[5]{{\bf else }: Send ``0'' to $\counter(e)$}
    \INDSTATE[4]{Let $f\ui = \hat f\ui$}
    \INDSTATE[3]{{\bf else}:}
    \INDSTATE[4]{{\bf for} each $e$: Send ``0'' to $\counter(e)$}

    \INDSTATE[1]{{\bf for} each $e$: Let $U_e$ be the output of $\counter(e)$}
    \STATE{{\bf Output:} $\{U_e\}_{e\in E}$}
    \end{algorithmic}
\end{algorithm}

We will now focus on analyzing the best-response dynamics within
$\br$. Note that in the analysis we might say ``player plays
best-response'' or ``player deviates''; while these sound natural, all
these procedures are simulated by the coordinator and the protocol is
non-interactive.

\begin{lemma}\label{translate}
At any moment of the dynamics, let $f\in \RR^{|E|}$ be the flow given
by all players' paths and let $g\in \RR^{|E|}$ be the count vector
given by the counters $\{\counter(e)\}_{e\in E}$. Suppose that player
$i$'s path $f\ui$ is an $\eta$-best-response with respect to $g$, then
the path $f\ui$ is an $(\eta + \lambda m r + \lambda m)$-best-response.
\end{lemma}

\begin{proof}
By Claim~\ref{counteracc}, we guarantee that throughout the dynamics,
for each $e\in E$
\[
 |\counter(e) - f| \leq r.
\]
This allows us to bound the cost difference of the same path with
respect to two flows $f$ and $g$ --- for any path $P\subseteq E$
\[
|c_P(f) - c_P(g)| \leq \lambda m r.
\]
This implies
\[
|\min_{P'\in \cA_i} c_{P'}(f) - \min_{P'\in \cA_i} c_{P'}(g)| \leq \lambda m r.
\]
Note that $c_{f\ui}(g) - \min_{P'\in \cA_i} c_{P'}(g) \leq \eta$ by
our assumption, then it follows from the last two inequalities that
\[
c_{f\ui}(f) \leq \min_{P'\in \cA_i} c_{P'}(f) + \eta + \lambda m r,
\]
so $f\ui$ is a $(\eta + \lambda m r)$-best-response w.r.t. the flow
$f$. Let $f'_i$ be a deviation of player $i$ and let $f' = (f'_i,
f^{(-i)})$ be the resulting flow. We know that for each edge $e$,
$|f'_e - f_e|\leq 1$. It follows that
\[
\min_{P'\in \cA_i} c_{P'}(f) \leq \min_{P'\in \cA_i} c_{P'}(f') + \lambda m.
\]
Therefore, $c_{f\ui}(f) \leq \min_{P'\in \cA_i} c_{P'}(f') + \eta + \lambda m
r + \lambda m$, which guarantees that $f\ui$ is a $(\eta + \lambda m
r + \lambda m)$-best-response.
\end{proof}

\begin{lemma}
Every time a player makes a deviation in the dynamics, the potential
function $\Psi$ decreases by at least $(\alpha - 2\lambda m r -
\lambda m)$.
\end{lemma}

\begin{proof}

Let $f$ denote the true flow among the $n$ players. Since the amount
the potential function decreases equals the amount that player
decreases her cost, we can bound $c_{f\ui}(f) - c_{\hat
  f\ui}(f')$, where $f' = (\hat f\ui, f^{(-i)})$.  Let $g$ denote the count
vector given by the counters $(\counter(e))_{e\in E}$.  Suppose a
player $i$ has her path switched from $f\ui$ to $\hat f\ui$ during the
dynamics. Then this means
\[
\min_{P\in \cA_i} c_P(g) = c_{\hat f\ui}(g) \leq c_{f\ui}(g) - \alpha.
\]
By the accuracy guarantee of Claim~\ref{counteracc},
\[
|c_{P}(f) - c_{P}(g)| \leq \lambda m r \qquad \mbox{ and } \qquad
\left|\min_{P'\in \cA_i} c_{P'}(f) - \min_{P'\in \cA_i} c_{P'}(g) \right| \leq \lambda m r.
\]
This means
\[
|c_{f\ui}(g) - c_{f\ui}(f)| \leq \lambda m r \qquad \mbox{and} \qquad
|c_{\hat f\ui}(g) - c_{\hat f\ui}(f)| \leq \lambda m r.
\]
Furthermore, note that $|f_e - f'_e| \leq 1$ for each edge $e$
since they differ by only player $i$'s path, so
\[
|c_{\hat f\ui}(f') - c_{\hat f\ui}(f)| \leq \lambda m.
\]
Combining all the inequalities above we get
\begin{align*}
c_{f\ui}(f) - c_{\hat f\ui}(f') &= (c_{f\ui}(f) - c_{f\ui}(g)) + (c_{f\ui}(g) - c_{\hat f\ui}(g))
\\ &+ (c_{\hat f\ui}(g) - c_{\hat f\ui}(f)) + (c_{\hat f\ui}(f) - c_{\hat f\ui}(f'))\\
&\geq - \lambda m r + \alpha - \lambda m r -\lambda m = \alpha - 2\lambda m r - \lambda m,
\end{align*}
which also lower bounds the amount that the potential function
decreases.
\end{proof}

\begin{lemma}\label{equilcondition}
At the end of the best-response dynamics, the players are playing
$(\alpha + \lambda m r + \lambda m)$-approximate equilibrium flow in the routing
game instance.
\end{lemma}

\begin{proof}
In each iteration, there is at least one player performing a
deviation, which will decrease the potential function by at least
$(\alpha - 2\lambda m r - \lambda m) = l$.

Note that the initial flow has potential
\[
\Psi(f) \equiv \sum_{e\in E}\sum_{i=1}^{f_e} c_e(i) \leq nm.
\]
This means after at most $T = mn/l$ iterations of the best-response
dynamics, every player is playing $\alpha$-best-response with respect
to the flow given by $(\counter(e))_{e\in E}$. By~\Cref{translate},
each player is playing a $(\alpha + \lambda m r + \lambda
m)$-best-response w.r.t. to the final flow $f$. Hence, the final flow
is an $(\alpha + \lambda m r + \lambda m)$-equilibrium flow.
\end{proof}

Given that we have shown that at the end of the best-response dynamics
in $\br$, the players are playing an approximate equilibrium, it only remains to construct a decoding function for the players to recover
their own sequence of actions in the dynamics. Observe that $\br$ outputs
the list of update steps across all counters, which allows each player
to simulate the history of the approximate counts. Therefore the
decoding function is straightforward: first call $\exstream$ to
extract the history of counts in the best-response dynamics; then each
player $i$ first forms her own initial flow by picking the shortest path in her action set,
and then at every time step $t$ such that $(t \mod n) \equiv i$, she
decides whether to switch to a best-response with respect to the counts. Since her best response, after breaking ties lexicographically, is uniquely determined, in this way she is able to determine which path she is playing along at the end of the dynamics, which is her part of the approximate equilibrium.  The
full description of the algorithm $\exflow$ is presented
in~\Cref{exflow}.

\begin{algorithm}[h]
\caption{$\exflow((s_i, d_i), \{U_e\}_{e\in E}, \alpha, r)$}
\label{exflow}
\begin{algorithmic}

  \STATE{\textbf{Input:} a player $i$'s source-destination pair $(s_i,
    d_i)$, message containing update steps $\{U_e\}_{e\in E}$ for the
    $\{\counter(e)\}_{e\in E}$, best-response parameter $\alpha$ and
    refinement parameter}
  \INDSTATE{{\bf Initialize:} $l = \alpha - 2\lambda m r - \lambda m$}
  \INDSTATE{{\bf for } each edge $e\in E$:}
  \INDSTATE[2]{Let $C_e = \exstream(U_e, r, T)$ be the history of approximate counts}
  \INDSTATE{\bf Form the initial flow:}
  \INDSTATE[2]{Let $f\ui$ be the $s_i$-$d_i$ path with the fewest number of edges, break ties lexicographically}
  \INDSTATE{{\bf for} $t = 1, \ldots , \frac{mn}{l}$:}
  \INDSTATE[2]{Let flow $g = (C_e(t n + i))_{e\in E}$}
  \INDSTATE[2]{{\bf if} $f\ui$ is not an $\alpha$-best-response w.r.t. $g$}
  \INDSTATE[3]{Switch $f\ui$ to a best-response w.r.t. $g$}
  \INDSTATE{{\bf Output:} the final $s_i$-$d_i$ path $f\ui$}
\end{algorithmic}
\end{algorithm}

\begin{claim}
Suppose that $\br((G, \{(s_i, d_i)\}_i, c), \alpha, r)$ outputs a set
of update step lists $\{U_e\}_{e\in E}$, then for each player $i$, her
output flow from the instantiation $\exflow((s_i, d_i), \{U_e\}_{e\in
  E}, \alpha , r)$ is the same as the final flow in the best-response
dynamics in $\br$.
\end{claim}

\begin{theorem}\label{flowmain}
Fix any $\eps > 2\lambda m$. Let $r = \frac{\eps - 2\lambda m
}{6\lambda m}$, and $\alpha = \eps - \lambda m r -\lambda m$. Then
given any routing game instance $\Gamma = (G, \{(s_i, d_i)\}_i, c)$,
the coordination protocol $(\br(\Gamma, \alpha, r), \exflow((s_i,
d_i), \{U_e\}_{e\in E}, \alpha , r))$ coordinates the players to play
an $\eps$-approximate equilibrium flow and has coordination complexity
of \[
\tilde O\left(\frac{ \lambda m^2 n}{(\eps - 2\lambda m)^2}\right).
\]
\end{theorem}

\begin{proof}
Given the parameters we choose, the players will end up playing
an $\eps$-approximate equilibrium by~\Cref{equilcondition}.

Next, we bound the encoding length of the output from $\br$. Recall
that $l = \alpha - 2\lambda m r - \lambda m$, and each time a player
makes a deviation in the best-response dynamics, the potential
function decreases by $l$. By the same analysis of~\Cref{equilcondition}, we know
that the total number of deviations that occur in the dynamics (across all
edges) is bounded $T = mn/l$. Since each counter updates its count
only if the count is changed by $r$, the total number of updates
across all counters is bounded by $m T/r$. Also, the total length of
each counter is bounded $(T+1)n$. By Claim~\ref{counteracc}, the output
list of update steps among all counters has encoding length at most
\[
O\left( \frac{m^2n}{(\eps - 3 \lambda m r - 2\lambda m)r} \log\left(\frac{m^2n^2}{l} \right) \right) = \tilde O \left( \frac{\lambda m^3 n}{(\eps - 2\lambda m)^2} \right)
\]
which recovers our stated bound.
\end{proof}

To interpret the bound in~\Cref{flowmain}, consider the case in which the routing game is a \emph{large game} -- in which $n$ is substantially larger than the number of edges $m$, and in which no player has a large influence on the latency of any single edge. In our context, this means the Lipschitz parameter
$\lambda$ of the cost function $c_e$ is small. Since the range of the
cost is normalized to lie between 0 and 1, it is reasonable to assume that $\lambda =
O(1/n)$. Given such a largeness assumption, the result
of~\Cref{flowmain} gives a coordination protocol that coordinates a
$(m/n)$-approximate equilibrium flow with a coordination complexity of
$\tilde O(m^3/\eps^2)$ for any constant $\eps$.

\begin{remark}
Another application for this coordination technique is the general
allocation problem when the players have \emph{gross substitutes}
preferences. In this problem, there are $n$ players and $m$ types of
goods, and each type of good has a supply of $s$. A coordinator who
knows the preferences of the players can first simulate $m$
simultaneous ascending price auctions, one for each type of good. The
analysis of~\cite{KC82} shows that the prices and the allocation in
the auctions converge to \emph{Walrasian equilibrium}: each buyer is
simultaneously able to buy his most preferred bundle of goods under
the set of prices. We could then have a similar coordination protocol:
compress an approximate version of the ascending auctions dynamics
using $\approxC$, and broadcast a message that allows each player to
reconstruct the price trajectory in the auctions, which can then
coordinate a high-welfare allocation.
\end{remark}

\pagebreak
\bibliographystyle{alpha}

\bibliography{billboard.bbl}

\appendix

\section{A Coordination Protocol For Many-to-One Stable Matchings}
Finally, we present a coordination protocol for the many-to-one stable
matchings problem, which is relatively straightforward. A many-to-one stable
matching problem consists of $k$ schools $U=\{u_1, \ldots , u_k\}$,
each with capacity $C_j$ and $n$ students $S=\{s_1, \ldots ,
s_n\}$. Every student $i$ has a strict preference ordering $\succ_i$
over all the schools, and each school $j$ has a strict preference
ordering $\succ_j$ over the students. It
will be useful for us to think of a school $u$'s ordering over
students $A$ as assigning a unique \emph{score} $\score_u(s) \in \{1,
\ldots , n\}$ to every student, in descending order (for example,
these could be student scores on an entrance exam). Therefore, each
student's private data is $D_i = (\succ_i, \{\score_u(s)\}_u)$.  We
recall the standard notion of stability in a many-to-one matching.

\begin{definition} A matching $\mu:S \rightarrow U\cup \emptyset$ is
  \emph{feasible} and \emph{stable} if:
\begin{enumerate}
\item (Feasibility) For each $u_j \in U$, $|\{i : \mu(a_i) = u_j\}|
\leq C_j$ \label{cond:feasibility}
\item (No Blocking Pairs with Filled Seats) For each $a_i \in A$, and
each $u_j \in U$ such that $\mu(a_i) \neq u_j$, either $\mu(a_i)
\succ_{a_i} u_j$ or for every student $a'_i \in \mu^{-1}(u_j)$, $a'_i
\succ_{u_j} a_i$. \label{cond:filledblocking}
\item (No Blocking Pairs with Empty Seats) For every $u_j \in U$ such
that $|\mu^{-1}(u_j)| < C_j$, and for every student $a_i \in A$ such
that $a_i \succ_{u_j} \emptyset$, $\mu(a_i) \succ_{a_i}
u_j$. \label{cond:emptyblocking}
\end{enumerate}
\end{definition}

A simple way to specify a matching is to specify an admissions threshold for every school -- i.e. a score $\adm(j)$ that represents the minimum score student the school is willing to accept. A set of admissions thresholds $\adm$ defines a matching $\mu_{\adm}$ in the natural way, in which every student enrolls in their most preferred school to which they have been admitted:

\[\mu_{\adm}(i) :=  \arg\max_{\succ_{a_i}} \{u_j
\mid \score_{u_j}(a_i) \geq \adm(j)\}.
\]

A set of admission scores is stable and feasible if the matching it induces is stable and feasible. 

Now we give a coordination protocol to coordinate students to select schools that form 
a stable matching. The protocol crucially relies on the Gale-Shapley
deferred acceptance algorithm --- the coordinator will first simulate
the dynamics of the deferred acceptance algorithm based on the student
profiles, and obtains the stable matching along with the associated
admission scores. Then it suffices for the planner to publish the list
of admission scores so that all the students can coordinate on the
stable matching. We present an score-based deferred acceptance
algorithm in~\Cref{stab}.

\begin{algorithm}[h]
\caption{Deferred Acceptance (with Admission Scores) $\stab(D)$}
 \label{stab}
  \begin{algorithmic}
    \STATE{\textbf{Input:} $n$ students' data $D$ including their
      preferences over the schools and score profiles}
    \INDSTATE{{\bf for} each school $u_j\in U$}
    \INDSTATE[2]{admission score $\adm(j) = n$; temporary enrolled students $\temp(j) = \emptyset$}
    \INDSTATE{{\bf for} each student $s_i \in S$: let $\mu(s_i) = \perp$}

    \INDSTATE{{\bf while} there is some under-enrolled school $u_{j'}$ such that $|\temp(j')| < C_{j'}$ and $\adm(j') > 1$}
    \INDSTATE[2]{$\adm(j') \leftarrow \adm(j') - 1$}
    \INDSTATE[2]{{\bf for } each student $s_i$:}
    \INDSTATE[3]{{\bf if}} $\mu(s_i) \neq \arg\max_{\succ_{a_i}} \{u_j \mid \score_{u_j}(a_i) \geq \adm(j)\}$
    \INDSTATE[4]{{\bf then} $\temp(j) \leftarrow \temp(j) \setminus \{a_i\}$}
    \INDSTATE[4]{$\mu(a_i) =  \arg\max_{\succ_{a_i}} \{u_j \mid \score_{u_j}(a_i) \geq \adm(j)\}$}
    \INDSTATE[4]{$\temp(\mu(a_i)) \leftarrow \temp(\mu(a_i)) \cup \{a_i\}$}
  \STATE{{\bf Output:} the final admission scores $\{\adm(j)\}$}
    \end{algorithmic}
\end{algorithm}

The following is a standard fact about the deferred acceptance
algorithm.

\begin{claim}
The final matching $\mu$ computed from $\stab(D)$ is a stable.
\end{claim}

Consider the following coordination protocol: the coordinator will
first run $\stab$ and broadcast the set of admission scores, and then
based on the score, each student will just enroll in the favorite
school that she qualifies, that is
 \[ \arg\max_{\succ_{a_i}} \{u_j
\mid \score_{u_j}(a_i) \geq \adm(j)\}.
\]
Note that the set of admission scores can be encoded with at most $O(k
\log(n))$ bits.

\begin{theorem}
  There exists a coordination protocol with coordination complexity of
  $O(k \log{n})$ that coordinates the students to coordinate on a
  stable matching.
\end{theorem}

\section{Missing Proofs from~\Cref{sec:lower}}
First, let us first fix some notations. Given any two random variables
$X$ and $Y$, we will write $I(X : Y)$ to denote the mutual
information, $\KL(X:Y)$ to denote the Kullback Leibler divergence, and
$\delta(X, Y)$ to denote the total variation distance between the
random variables. We denote the Shannon entropy of a random variable
$Z$ by $H(Z)$.

\ranindex*

\begin{proof}
Fix a protocol where the coordinator broadcasts $\ell$ bits in the worst 
case. We will show that $\ell$ must be large.
Let $M(I)$ be the (random) message that Alice sends to Bob based on
her input $I$. Let $S(I) = \langle S_1, S_2, \ldots, S_t \rangle$ and
$u(I)=\langle u_1,u_2,\ldots,u_t \rangle$.  Then, $S(I)$, $u(I)$ and
$M(I)$ are random variables. Fix a value $S$ for $S(I)$ and so that
conditioned on ``$S(I)=S$'' the the probability of success is at least
$p$ (such an $S$ must exist). Then,
\[
I[u(I) :  M(I)] = H[M(I)  - H(M(I) \mid u(I)) \leq H(M(I)) \leq \Expectation[|M(I)|]
\leq \ell.
\]
On the other hand, since $u_1, u_2, \ldots, u_t$ are independent (conditioned
on ``$S(I)=S$''), we have
\[
\ell \geq I[u(I) : M(I)] \geq I[u_1 : M(I)] + \cdots + I[u_t : M(I)],
\]
implying that
\[
\Expectation_{i} \left[  I[u_i : M(I)]\right] \leq \frac{\ell}{t}.
\]
Let $u_{iz}$ denote the random variable whose distribution is the same
as that of $u_i$ when conditioned on the event $M(I)=z$, and let the
distance between the distributions of $u_{iz}$ and $u_i$ be
\[ \delta(u_{iz}, u_i) := \sum_{w \in S_i} |\Pr[u_{iz}=w] -\Pr[u_i=w]|.\]
Then, we have
\begin{align*}
I[u_i : M(I)] =& \Expectation_z \left[\KL(u_{iz} : u_i) \right] & \mbox{(from definition)}\\
\geq & (2 \log e) \Expectation_z \left[\left(\delta(u_{iz}, u_i)\right)^2 \right] & \mbox{(by Pinsker's inequality)} \\
\geq & (2\log e) \left(\Expectation_z\left[\delta(u_{iz} , u_i) \right]\right)^2.
& \mbox{(by Jensen's inequality)}
\end{align*}
With one more
application of Jensen's inequality, we obtain the following
\[
\Expectation_{i} \left[I[u_i : M(I)]\right] \geq
(2\log e) \Expectation_{i} \left[ \left(\Expectation_z\left[\delta(u_{iz} : u_i)
    \right] \right)^2 \right] \geq
(2\log e) \left(\Expectation_{i} \left[ \Expectation_z\left[\delta(u_{iz} : u_i)
    \right] \right]\right)^2.
\]
Suppose Bob succeeds in guessing the special element of $S_i$ with
probability $p_i$; then $\Expectation_i[p_i] = p$. 
Furthermore, we have $\delta(u_{iz}, u_i) \geq
2|p_i- 1/k|$. Thus, by Jensen's inequality, we have
\[
\left(\Expectation_{i} \left[ \Expectation_z\left[\delta(u_{iz} : u_i)
    \right] \right]\right)^2 \geq 4(p - 1/k)^2
\]
It follows that $\ell /t \geq (8\log e) (p-1/k)^2$, that is, $\ell
\geq (8 \log e) t(p-1/k)^2$.
\end{proof}

\sampledmatching*
\begin{proof}
For $w \in W'$, let $d_w$ be the number of copies of good $w$ that
have been assigned in the matching $M^*$. Then $\sum_w d_w \geq
b\OPT'/\rho$.  If $d_w \geq 1$, then we have the following
estimate for the probability that $w$ is chosen.
\begin{align}
\Pr[\mbox{$w$ is chosen}]
&\geq d_w \frac{1}{b} \left(1 - \frac{1}{b}\right)^{d_w -1}\\
& = \frac{d_w}{b} \left(1 + \frac{1}{b-1}\right)^{-(d_w -1)}\\
&\geq \frac{d_w}{b} \exp(-(d_w-1)/(b-1)) & \mbox{(because $1+x \leq e^x$)}\\
&\geq \frac{d_w}{eb}. & \mbox{(because $d_w \leq b$)}
\end{align}
Let $M'$ be the sampled matching in $G'$. Then by linearity of
expectation:
\[
\Expectation[|M^*|]  \geq \sum_{w \in W'} \frac{d_w}{eb}\\
\geq \frac{b \OPT'}{e b \rho} \geq \frac{\OPT'}{3\rho}.
\]
This completes our proof.
\end{proof}

\jr{I have tried to avoid some applications of Markov's
  inequality. Please check}

\section{Missing Proofs from~\Cref{sec:convex}}

\dualprimal*

\begin{proof}

Fixing any $x$, the function $\cL(x, \lambda)$ is $\sqrt nk$-Lipschitz
with respect to $\ell_2$ norm because for any $\lambda, \lambda'$
\[
\|\cL(x, \lambda) - \cL(x, \lambda')\| = \left\|\sum_{j=1}^k  (\lambda_j - \lambda_j') \sum_{i=1}^n c_j\ui\left( x\ui \right)  \right\| \leq \|\lambda - \lambda'\| \left\|\sum_{i=1}^n c_j\ui\left( x\ui \right)\right\| \leq
 \|\lambda -\lambda'\| \sqrt{nk}
\]
By the property of Lipschitz functions we know that $g(\lambda)
= \max_{x\in \cF} \cL(x, \lambda)$ is also
$\sqrt{nk}$-Lipschitz.\sw{cite!} Since $\|\lambda\ub
- \hat \lambda\| \leq \alpha$, we can then bound
\[
\|\cL(x\ub, \hat \lambda) - \cL(x\ub, \lambda\ub)\| \leq \alpha \sqrt{nk},
\]
and also
\[
\|\cL(x\ub ,  \lambda\ub) - \cL(\hat x, \hat \lambda)\| = \|g(\lambda\ub) - g(\hat\lambda)\| \leq \alpha \sqrt{nk}.
\]
It follows that \[
\|\cL(x\ub, \hat\lambda) - \cL(\hat
x, \hat\lambda)\|\leq 2\alpha\sqrt nk.\] Note that fixing
$\hat\lambda$, the function $\cL(x, \hat\lambda)$ is $\eta$-strongly
concave in $x$. Since $\hat x$ is a maximizer of $\cL(\cdot, \hat\lambda)$, we have
\[
\|\hat x - x\ub \|^2 \leq \frac{2}{\eta}\left( \cL(\hat x,\hat \lambda) - \cL(x\ub, \hat\lambda) \right) \leq 4\alpha\sqrt nk / \eta
\]
Therefore, we must have $\|\hat x - x\ub\| \leq \frac{2 \sqrt{\alpha}
(nk)^{1/4}}{\sqrt \eta}$.
\end{proof}

\matchingviolation*
\begin{proof}
Based on the relation between $\ell_1$ and $\ell_2$ norm, we have
\[
\min_{x\in \cF} \|x - \hat x\|_1 \leq \sqrt{nk}\min_{x\in \cF} \|x - \hat x\| \leq
\sqrt{nk} \eps
\]
This means
\[
\sum_{j=1}^k \left( \sum_{i=1}^n \hat x_{i,j} - b_j \right)_+ \leq \sqrt{nk}\eps
\]
Note that for each good $j$ and any $\delta \in (0,1)$, we have from
Chernoff-Hoeffding bound that
\[
\Pr\left[\sum_{i=1}^n x'_{i,j}  > (1+\delta) \sum_{i=1}^n \hat x_{i,j}\right] <
\exp\left( -\delta^2 n / 3\right).
\]
Let $X_j = \sum_{i=1}^n \hat x_{i,j}$. If we set $\beta / k =
\exp\left( -\delta^2 n / 3\right)$, then by union bound we have the
following except with probability $\beta$
\[
\mbox{for all }j \in [k], \qquad \sum_{i=1}^n x'_{i,j} \leq  \left(1 + \sqrt{3\log(k/\beta)/X_j}\right) X_j
\]
Also, $v(\hat x) = \|\hat x\|_1$ because for any $(i,j)$ such that
$v_{i,j} = 0$, we must have $\hat x_{i,j} = 0$, otherwise the player
$i$ could increase the regularized objective value by having $\hat
x_{i,j} = 0$ in~\Cref{br}. It follows that
\[
\sum_{j=1}^k\left(\sum_{i=1}^n x'_{i,j} - X_j\right) \leq \sum_{j=1}^k \sqrt{3\log(k/\beta) X_j} \leq \sqrt{3k\log(k/\beta) \|\hat x\|_1} = \sqrt{3k\log(k/\beta)\hat V}
\]
Therefore,
\[
\sum_{j=1}^k \left( \sum_{i=1}^n x'_{i,j} - b_j \right)_+ \leq \sqrt{nk}\eps + \sqrt{3k \log(k/\beta) \hat V},
\]
which recovers the stated bound.
\end{proof}

\mainconvex*

\begin{proof}
We instantiate our coordination mechanism for linearly separable convex programs with $\eta = \eps =
1/(100n^3k^3)$, and round the solution to $x'$.  Applying our previous lemmas, we get that with probability at least $1- \beta$:
\begin{align*}
\sum_{j=1}^k \min\left\{\sum_{i=1}^n v_{i,j} x'_{i,j}, b_j\right\}
&\geq \hat V - \log(4/\beta)\sqrt{\hat V} - \sqrt{nk}\eps - \sqrt{3k\log(2k/\beta)\hat V}\\
&\geq \OPT - n(\eps + \eta)- \log(4/\beta)\sqrt{\hat V} - \sqrt{nk}\eps - \sqrt{3k\log(2k/\beta)\hat V}\\
&\geq \OPT - n(\eps + \eta)- \left(\log(4/\beta) + \sqrt{3k\log(2k/\beta)}\right)\sqrt{\OPT + n(\eps + \eta)} - \sqrt{nk}\eps \\
&\geq \OPT - n(\eps + \eta) - 4\sqrt{k}\log(2k/\beta) \sqrt{\OPT + n(\eps + \eta)} - \sqrt{nk}\eps\\
&\geq \OPT - 8\sqrt{k}\log(2k/\beta) \sqrt{\OPT}
\end{align*}
which recovers the stated bound. Note that the coordination complexity
of this mechanism is $O(k\log(nk))$ by~\Cref{noconstants}.
\end{proof}

\end{document}